\documentclass[11pt]{amsart}

\usepackage{amsmath}
\usepackage{amssymb}
\usepackage[final]{graphicx}

\newtheorem{thm}{Theorem}
\newtheorem{lem}[thm]{Lemma}
\newtheorem{prop}[thm]{Proposition}

\theoremstyle{definition}
\newtheorem{defn}{Definition}[section]

\newtheorem{notn}{Notation}

\theoremstyle{remark}
\newtheorem{remark}{Remark}[thm] 

\theoremstyle{plain}

\numberwithin{equation}{section}

\def\CC{{\mathbb C}}

\def\HH{{\mathbb H}}
\def\NN{{\mathbb N}}

\def\RR{{\mathbb R}}

\def\i{\mathrm{i}}

\def\id{\operatorname{id}}

\def\D{\operatorname{D{}}}

\def\G{\operatorname{G{}}}

\def\SL{\operatorname{SL}}

\def\PSL{\operatorname{PSL}}

\def\Area{\operatorname{Area}}
\def\ord{\operatorname{ord}}

\def\scrA{{\mathcal A}}

\def\scrF{{\mathcal F}}

\def\scrI{{\mathcal I}}

\def\scrL{{\mathcal L}}

\def\scrM{{\mathcal M}}

\def\Re{\operatorname{Re}}
\def\Im{\operatorname{Im}}

\def\Dom{\operatorname{Dom}}

\newcommand{\be}{\begin{equation}}
\newcommand{\ee}{\end{equation}}
\newcommand{\ba}{\begin{align}}
\newcommand{\ea}{\end{align}}

\begin{document}

\begin{center}
\title[The trace formula for a point scatterer]{The trace formula for a point scatterer \\ on a compact hyperbolic surface}
\author{Henrik Uebersch\"ar}
\address{The Raymond and Beverly Sackler Schoool of Mathematical Sciences, Tel Aviv University, Tel Aviv 69978,
Israel.}
\email{henrik@post.tau.ac.il}
\date{\today. This work is part of the author's PhD thesis completed at the University of Bristol in January 2011. Support from an EPSRC doctoral training grant is gratefully acknowledged. The author is currently supported by a Minerva Fellowship.}
\maketitle
\end{center}

\begin{abstract}
An exact trace formula for the perturbation of the Laplacian by a Dirac delta potential on a compact hyperbolic Riemann surface is derived. The formula can be considered an analogue of the Selberg trace formula. The difference of perturbed and unperturbed trace is expressed as an identity term plus a sum over combinations of diffractive orbits which visit the position of the potential. 
\end{abstract}

\section{Introduction}

The Selberg trace formula \cite{S} is a central tool in the spectral theory of automorphic forms \cite{Hj2,Hj3,Iw} and has important applications in number theory \cite{Hj1}. It also plays a central role in the theory of quantum chaos as an exact analogue of Gutzwiller's celebrated trace formula \cite{Bz,CRR,G,PU} which links the distribution of the energy levels of a classically chaotic quantum system to the actions of periodic orbits. 

In quantum mechanics perturbations of self-adjoint Hamiltonians by Dirac delta potentials have long been studied \cite{Z} and serve as a simplified model for the interaction of particles in a quantum system \cite{BF,JSS,Bl}. In the context of Quantum Chaos Seba \cite{Sb} studied the perturbation of the Laplacian by a delta potential on a rectangle with irrational side ratio and Dirichlet boundary conditions. The classical dynamics of such a system - known as a Seba billiard - is integrable and identical to the dynamics on the free rectangular billiard. The quantum waves, however, are scattered by the delta potential. 

Seba observed that the eigenvalues of the perturbed system seem to repell each other. This phenomenon is also observed for quantum systems which are classically chaotic, where generically the distribution of the eigenvalue spacings is believed to be governed by random matrix theory \cite{Bh}. Bogomolny, Gerland and Schmit \cite{Bo} have investigated singular statistics and in the case of the Seba billiard they argue that the distribution of the eigenvalue spacings is not governed by random matrix theory, but in fact seems to be closer to semi-Poisson statistics. The true distribution, however, still remains unknown.

The impact of a delta potential on the spectrum of the Laplacian on hyperbolic surfaces is even more mysterious. Whereas the spectrum of the Laplacian on a rectangle is fairly well understood, much less is known about the spectrum for a hyperbolic surface. We hope that the development of an exact trace formula will shed some light on this interesting problem. 

Hillairet \cite{Hi} proved a semi-classical trace formula for a delta potential on a 3-dimensional Riemannian manifold. There is also a general trace formula for rank-one perturbations due to Krein \cite{Kr} which, however, does not give any information about diffractive orbits. Our work is also related to Venkov's proof \cite{V} of an analogue of the Selberg trace formula for an automorphic Schr\"odinger operator with a continuous nonnegative potential. 

\section{Set Up and Results}

\subsection{Set Up}
Let $\scrM$ be a compact hyperbolic surface. We can represent $\scrM$ as a quotient $\Gamma\backslash\HH$ of the hyperbolic halfplane $\HH$ by a co-compact lattice $\Gamma$ in the group of isometries $\SL(2,\RR)/\lbrace\pm\id\rbrace$. A brief review of hyperbolic surfaces and automorphic functions is given in sections 3.1 and 3.2.

Fix a point $z_{0}\in\scrM$. We denote by $\Delta$ the Laplacian on $\scrM$. Consider the operator $-\Delta_{\alpha,z_{0}}=-\Delta-\alpha(\delta_{z_{0}},\cdot)\delta_{z_{0}}$, $\alpha\in\RR\setminus\lbrace0\rbrace$. Define  $$D_{0}=C^\infty_0(\scrM\setminus\lbrace z_0\rbrace).$$ The operator $-\Delta_{\alpha,z_{0}}$ can be realised as a self-adjoint extension $-\Delta_{\varphi(\alpha)}$, $\varphi(\alpha)\in(-\pi,\pi)$, of the restricted operator $-\Delta_{\alpha,z_{0}}|_{D_{0}}=-\Delta|_{D_{0}}$ which is a positive symmetric operator with deficiency indices (1,1). For background reading on self-adjoint extensions of positive symmetric operators see \cite{Si}, Section X.1. For a discussion of the mathematical theory of rank-one perturbations of differential operators and their realisation as self-adjoint extensions see \cite{Al}, Chapter 1. We discuss the self-adjoint extensions $-\Delta_{\varphi(\alpha)}$ in detail in section 3.3. 

Since $\scrM$ is compact, the spectrum of the positive Laplacian on the surface is discrete. An eigenfunction $\varphi_{j}$, $j=0,1,2,\cdots$, satisfies $-\Delta\varphi_{j}=\lambda_{j}\varphi_{j}$, where $\lambda_{j}\geq0$ is the corresponding eigenvalue, and the spectrum accumulates at infinity: $$\lambda_{0}=0<\lambda_{1}\leq\cdots\leq\lambda_{n}\leq\cdots\to\infty$$

We denote the eigenfunctions of the perturbed Laplacian by $\varphi^{\alpha}_{j}$, $j=0,1,2,\cdots$, and they satisfy $-\Delta_{\varphi(\alpha)}\varphi^{\alpha}_{j}=\lambda^{\alpha}_{j}\varphi^{\alpha}_{j}$. The spectrum of $-\Delta_{\varphi(\alpha)}$ interlaces with the spectrum of $-\Delta$. For a generic surface (cf. \cite{CdV}, remark after Thm. 2), where all eigenvalues have multiplicity $1$ and $\varphi_j(z_0)\neq0$ for all $j$, we have
\begin{equation}
\lambda_{0}^{\alpha}<0=\lambda_{0}<\lambda_{1}^{\alpha}<\lambda_{1}<\cdots<\lambda_{n-1}<\lambda_{n}^{\alpha}
<\lambda_{n}<\cdots\to\infty.
\end{equation}

With regard to the multiplicity of an eigenvalue $\lambda_{j}$ and the associated perturbed eigenvalue $\lambda_{j}^{\alpha}$ we distinguish the following cases (cf. \cite{CdV}, Thm. 2). Denote the associated eigenspaces by $E_{\lambda_{j}}$ and $E_{\lambda^{\alpha}_{j}}$. 

\begin{itemize}
\item[(a1)] $\dim E_{\lambda^{\alpha}_{j}}=\dim E_{\lambda_{j}}=n>1$ if all eigenfunctions in $E_{\lambda_{j}}$ vanish at $z_{0}$,

\item[(a2)] $\dim E_{\lambda^{\alpha}_{j}}=\dim E_{\lambda_{j}}-1=n-1>0$ if at least one eigenfunction in $E_{\lambda_{j}}$ does not vanish at $z_{0}$, the perturbed eigenfunctions are superpositions of two unperturbed eigenfunctions and vanish at $z_{0}$

\item[(b)] $\dim E_{\lambda^{\alpha}_{j}}=\dim E_{\lambda_{j}}=1$ if $\varphi_{j}(z_{0})\neq0$, the associated perturbed eigenfunction is a Green function $G_{\lambda^{\alpha}_{j}}(\cdot,z_{0})$.
\end{itemize}

We write an eigenvalue of $-\Delta$ as $\lambda_{j}=\tfrac{1}{4}+\rho_{j}^{2}$, $\rho_{j}\in\RR_{+}\cup\i[0,\tfrac{1}{2}]$, and an eigenvalue of $-\Delta_{\varphi(\alpha)}$ as $\lambda^{\alpha}_{j}=\tfrac{1}{4}+(\rho^{\alpha}_{j})^{2}$, $\rho^\alpha_{j}\in\RR_{+}\cup\i\RR_{+}$. Denote by $\lbrace\mu_j\rbrace_{j=0}^\infty$ the Laplacian eigenvalues ignoring multiplicities. Denote by $\lbrace\mu^\alpha_j\rbrace_{j=0}^\infty$ the perturbed eigenvalues of type (b). We refer to them as new eigenvalues to distinguish them from the perturbed eigenvalues of types (a1) and (a2). We write $\mu_j=\tfrac{1}{4}+r_j^2$ and $\mu^\alpha_j=\tfrac{1}{4}+(r^\alpha_j)^2$. The new eigenvalues $\lbrace\mu^\alpha_j\rbrace$ are determined from the zeros of a certain meromorphic function (cf. p. 9, Prop. 3). As we will see later (cf. p. 13, Thm. 6) there is exactly one new eigenvalue in between two consecutive old eigenvalues.

We introduce the following space of test functions.
\begin{defn}
Let $\sigma\geq\tfrac{1}{2}$ and $\delta>0$. We define $H_{\sigma,\,\delta}$ to be the space of functions $h:\CC\to\CC$, s.t
\begin{enumerate}
\item[(i)] $h$ is even,
\item[(ii)] $h$ is analytic in the strip $\left|\Im{\rho}\right|\leq\sigma$,
\item[(iii)] $h(\rho)\ll(1+|\Re\rho|)^{-2-\delta}$ uniformly in the strip $\left|\Im\rho\right|\leq\sigma$.
\end{enumerate}
\end{defn}

\begin{remark}
It follows from the observations above about the multiplicities that we have the identity
\begin{equation}\label{specid}
\sum_{j=0}^\infty\lbrace h(r_j^\alpha)-h(r_j)\rbrace=\sum_{j=0}^\infty\lbrace h(\rho_j^\alpha)-h(\rho_j)\rbrace.
\end{equation}
where the sums are absolutely convergent because of the decay of $h$ and Weyl's law.
\end{remark}

\subsection{Results}
Let $m_{\Gamma}=|\scrI|$, where
\begin{equation}
\scrI=\lbrace\gamma\in\Gamma\mid\gamma z_{0}=z_{0}\rbrace.
\end{equation}
Throughout this paper we denote $\psi(s)=\tfrac{1}{2\pi}\Gamma'(s)/\Gamma(s)$, where $\Gamma(s)$ is the usual Gamma function.

Let $\nu>v_{\beta}$ if $1+m_{\Gamma}\beta\psi(\tfrac{1}{2}+\i\rho)$ has a zero $-\i v_{\beta}$ in the interval $(0,-\i\sigma)$, and $\nu=0$ otherwise. In fact there is at most one zero of $1+m_{\Gamma}\beta\psi(\tfrac{1}{2}+\i\rho)$ in the halfplane $\Im\rho<0$ and it lies on the imaginary axis. To see this consider the representation (B9) in \cite{Iw} (beware of our additional factor $\tfrac{1}{2\pi}$). It follows that
\begin{equation}
\Im(1+m_{\Gamma}\beta\psi(\tfrac{1}{2}+\i\rho))=m_{\Gamma}\beta\frac{\Re\rho}{2\pi}\sum_{n=0}^{\infty}|n+\tfrac{1}{2}+\i\rho|^{-2},
\end{equation}
so $1+m_{\Gamma}\beta\psi(\tfrac{1}{2}+\i\rho)$ has no zeros off the imaginary line. For $\rho=-\i v$ and $v>\tfrac{1}{2}$ we have $1+m_{\Gamma}\beta\psi(\tfrac{1}{2}+v)\in\RR^{+}$ and
\begin{equation}
\frac{d}{dv}(1+m_{\Gamma}\beta\psi(\tfrac{1}{2}+v))=\frac{m_{\Gamma}\beta}{2\pi}\sum_{n=0}^{\infty}(n+\tfrac{1}{2}+v)^{-2}>0.
\end{equation}
This together with the asymptotics $\Gamma'(s)/\Gamma(s)=\log s+O(|s|^{-1})$ for $\Re s>\tfrac{1}{2}$ (cf. (B11), pp. 198-9 in \cite{Iw}) and the observation that $\lim_{v\to-1/2^{+}}\psi(\tfrac{1}{2}+v)=-\infty$ implies that $1+m_{\Gamma}\beta\psi(\tfrac{1}{2}+\i\rho)$ has exactly one zero in the halfplane $\Im\rho<\tfrac{1}{2}$ at $\rho=-\i v_{\beta}$, $v_{\beta}>-\tfrac{1}{2}$.

\begin{defn}\label{transform}
For $h\in H_{\sigma,\delta}$, $\beta\in\RR$, $k\in\NN$ and $\nu$ as above, we define the following integral transform of $h$
\begin{equation}\label{trans}
g_{\beta,k}(t)=\frac{(-1)^{k}}{2\pi\i k}\int_{-\i\nu-\infty}^{-\i\nu+\infty}\frac{h'(\rho)e^{-\i\rho t}d\rho}{(1+m_{\Gamma}\beta\psi(\tfrac{1}{2}+\i\rho))^{k}}.
\end{equation}
\end{defn}

Selberg's trace formula relates the trace of the Laplacian to a sum over periodic orbits on $\Gamma\backslash\HH$. Our trace formula relates the difference between the traces of $-\Delta_{\varphi(\alpha)}$ and $-\Delta$ to so-called diffractive orbits. By a diffractive orbit associated with a group element $\gamma\in\Gamma$, which does not fix $z_{0}$, we mean the following. We consider the unique geodesic which connects $z_{0}$ to $\gamma z_{0}$ in $\HH$ and project it on to the surface $\Gamma\backslash\HH$. We hence obtain an orbit which starts and returns to $z_{0}$, which is, however, not necessarily a periodic orbit. We denote the length $d(\gamma z_{0},z_{0})$ of such an orbit by $l_{\gamma,z_{0}}$.

The following trace formula for compact surfaces is our main result. We prove it in section 5.
\begin{thm}\label{thm1}
Suppose $\scrM$ is a compact hyperbolic surface. Fix a point $z_{0}\in\Gamma\backslash\HH$ and $\alpha\in\RR\setminus\lbrace-1/c_{0}\rbrace\cup\lbrace\pm\infty\rbrace$, where $c_{0}\neq0$ is a certain real constant defined in \eqref{const0}. Let $-\Delta\varphi_{j}=\lambda_{j}\varphi_{j}$, and  $-\Delta_{\varphi(\alpha)}\varphi^{\alpha}_{j}=\lambda_{j}^{\alpha}\varphi^{\alpha}_{j}$, where $\lambda_{j}=\tfrac{1}{4}+\rho_{j}^{2}$, and $\lambda_{j}^{\alpha}=\tfrac{1}{4}+(\rho_{j}^{\alpha})^{2}$. For a certain positive constant $C(\Gamma,\alpha,z_{0})$ choose $\sigma>\tfrac{1}{2}$ such that 
\begin{equation}\label{condition}
(\tfrac{1}{2}+\sigma)^{1/2}\log(\tfrac{1}{2}+\sigma)>C(\Gamma,\alpha,z_{0})
\end{equation}
and let 
\begin{equation}\label{renorm}
\beta=\beta(\alpha)=
\begin{cases}
\alpha/(1+c_{0}\alpha),\quad \alpha\in\RR\setminus\lbrace-1/c_{0}\rbrace\\
1/c_{0},\quad \alpha=\pm\infty.
\end{cases}
\end{equation} 
Let $\nu=\nu(\alpha)$ be defined as in Definition \ref{transform}. For any $\delta>0$ and $h\in H_{\sigma,\delta}$ we have the identity
\begin{equation}
\begin{split}
&\sum_{j=0}^{\infty}\lbrace h(\rho^{\alpha}_{j})-h(\rho_{j})\rbrace\\
=\;&\frac{1}{2\pi}\int_{-\i\nu-\infty}^{-\i\nu+\infty}h(\rho)\frac{m_{\Gamma}\beta\psi'(\tfrac{1}{2}+\i\rho)}{1+m_{\Gamma}\beta\psi(\tfrac{1}{2}+\i\rho)}d\rho\\
&+\sum_{k=1}^{\infty}\beta^{k}\sum_{\gamma_{1},\cdots,\gamma_{k}\in\Gamma\backslash\scrI}
\int_{l_{\gamma_{1},z_{0}}}^{\infty}\cdots\int_{l_{\gamma_{k},z_{0}}}^{\infty}\frac{g_{\beta,k}(t_{1}+\cdots+t_{n})\prod_{n=1}^{k}dt_{n}}{\prod_{n=1}^{k}\sqrt{\cosh t_{n}-\cosh l_{\gamma_{n},z_{0}}}}.
\end{split}
\end{equation}
\end{thm}
\begin{remark}
The constant $C(\Gamma,\alpha,z_{0})$ arises naturally in the proof. The condition \eqref{condition} ensures that the series over the $k$-tuples of diffractive orbits converges absolutely.
\end{remark}
\begin{remark}
The relation \eqref{renorm} corresponds to a renormalisation of the coupling constant $\alpha$. For a physical interpretation of this renormalisation see for instance \cite{Bo}, bottom paragraph on p. 3.
\end{remark}

\section{Background}

\subsection{Hyperbolic surfaces}
We define the upper half-plane
\begin{equation}
\HH=\left\{x+\i y\in\CC\mid y>0\right\}
\end{equation}
with Riemannian metric $ds^{2}=y^{-2}(dx^{2}+dy^{2})$ and volume element $d\mu(z)=y^{-2}dx\,dy$, where $z=x+\i y$. The geodesic distance on $\HH$ is given by $d(\cdot,\cdot):\HH\times\HH\to\RR_{+}$, where
\begin{equation}
\cosh d(z,w)=1+\frac{|z-w|^{2}}{2\Im z \Im w}.
\end{equation}
The Laplacian on $\HH$ is of the form
\begin{equation}
\Delta=y^{2}\left(\frac{\partial^{2}}{\partial x^{2}}+\frac{\partial^{2}}{\partial y^{2}}\right).
\end{equation}
Define
\begin{equation}
\SL(2,\RR)=\left\{
\begin{pmatrix}
   a & b\\
   c & d
   \end{pmatrix}\Bigg|\quad
   a, b, c, d \in\RR,\quad ad-bc=1
\right\},
\end{equation}
and
\begin{equation}
\PSL(2,\RR)=\SL(2,\RR)/\lbrace\id,-\id\rbrace.
\end{equation}
The orientation-preserving isometries of $\HH$ are the fractional linear transformations
\begin{equation}
	 z\to\frac{az+b}{cz+d},\qquad
	 \begin{pmatrix}
   a & b\\
   c & d
   \end{pmatrix}
   \in\PSL(2,\RR).
\end{equation}

A hyperbolic Riemann surface $\scrM$ of finite volume can be represented as a quotient $\Gamma\backslash\HH$, where $\Gamma\subset\PSL(2,\RR)$ is a Fuchsian group of the first kind. This means that for any $z\in\HH$ the orbit $\Gamma z$ has no limit point in $\HH$. 

\subsection{Automorphic functions}

We define the space of automorphic functions 
\begin{equation}
\scrA(\Gamma\backslash\HH)=\left\{f:\HH\to\CC\mid\forall\gamma\in\Gamma:f(\gamma z)=f(z)\right\}.
\end{equation}
Let $\scrF$ be a connected fundamental domain for $\Gamma$. Let $f\in\scrA(\Gamma\backslash\HH)$ be a measurable function. We introduce the norm
\begin{equation}
\left\|f\right\|=\left(\int_{\scrF}|f(z)|^{2}d\mu(z)\right)^{1/2}
\end{equation}
and the space of square-integrable functions
\begin{equation}
L^{2}(\Gamma\backslash\HH)=\left\{f\in\scrA(\Gamma\backslash\HH)\mid\left\|f\right\|<+\infty\right\}.
\end{equation}
Let $f,g\in\scrA(\Gamma\backslash\HH)$ be measurable functions. We define an inner product on $L^{2}(\Gamma\backslash\HH)$ by
\begin{equation}
(f,g)=\int_{\scrF}f(z)\overline{g(z)}d\mu(z).
\end{equation}

An important example of an automorphic function is the automorphic Green function $G^{\Gamma}_{s}(\cdot,z_{0})$ which is defined by the method of images for $\Re s>1$
\begin{equation}
G^{\Gamma}_{s}(z,z_{0})=\sum_{\gamma\in\Gamma}G_{s}(z,\gamma z_{0}),\qquad z\neq z_{0} \mod \Gamma,
\end{equation}
where $G_{s}(\cdot,w)$ denotes the Green function on $\HH$. The automorphic Green function satisfies
\begin{equation}
(\Delta+s(1-s))G^{\Gamma}_{s}(\cdot,z_{0})=\delta_{z_{0}}.
\end{equation}
The function $G^{\Gamma}_{t}(\cdot,z_{0})$ is automorphic in the sense that for any $z\notin\Gamma z_{0}$
\begin{equation}
G^{\Gamma}_{t}(\gamma z,z_{0})=G^{\Gamma}_{t}(z,z_{0}).
\end{equation}
Since $\Gamma\backslash\HH$ is compact, we have
\begin{equation}
G^{\Gamma}_{1-s}(z,z_{0})=G^{\Gamma}_{s}(z,z_{0})
\end{equation}
which follows from the spectral expansion (in the distributional sense)
\begin{equation}\label{greenspec}
G^{\Gamma}_{s}(z,z_{0})=\sum_{j=0}^{\infty}\frac{\varphi_{j}(z)\overline{\varphi_{j}(z_{0})}}{\lambda_{j}-s(1-s)},
\qquad z\neq z_{0} \mod \Gamma,
\end{equation}
where $\lbrace\lambda_{j}\rbrace_{j=0}^{\infty}$ denotes the eigenvalues of the Laplacian. 

\subsection{Delta potentials}

Let $z_{0}\in\scrM$. We consider the one-parameter family of rank-one perturbations $\Delta_{\alpha,z_{0}}=\Delta+\alpha\delta_{z_{0}}(\delta_{z_{0}},\cdot)$, $\alpha\in\RR$, where $\delta_{z_{0}}$ denotes the Dirac functional at $z_{0}$. It is known (cf. \cite{CdV}, section 1) that the formal operator $-\Delta_{\alpha,z_{0}}$ can be realised as a self-adjoint extension of $-\Delta$ defined on $D_{0}=C^{\infty}_0(\scrM\setminus\lbrace z_0\rbrace)$.

It can easily be checked that $-\Delta:D_{0}\to L^{2}(\scrM)$ is symmetric with respect to $(\cdot,\cdot)$. Also $-\Delta_{\alpha,z_{0}}|_{D_{0}}=-\Delta|_{D_{0}}$ for any $\alpha\in\RR$. It is well-known (cf. for instance \cite{CdV}) that in dimension 2 the symmetric operator $-\Delta|_{D_{0}}$ has deficiency indices $(1,1)$. The deficiency elements are the automorphic Green function $\G^{\Gamma}_{\,t}(z,z_{0})$ and its conjugate $\G^{\Gamma}_{\,\bar{t}}(z,z_{0})$, where $t(1-t)=\i$ and $\Re t>\tfrac{1}{2}$. The Green function $\G^{\Gamma}_{\,t}(\cdot,z_{0})$ uniquely satisfies
\begin{equation}
(\Delta+t(1-t))G^{\Gamma}_{\,t}(\cdot,z_{0})=\delta_{z_{0}},
\end{equation}
and has the asymptotics
\begin{equation}\label{GAsympt}
G^{\Gamma}_{\,t}(z,z_{0})=\frac{m}{2\pi}\log d(z,z_{0})+\frac{\gamma}{2\pi}+\psi(t)+o(1),\qquad z\to z_{0},
\end{equation}
where $\gamma$ is Euler's constant, $m=\ord(z_{0})$ if $z_{0}$ corresponds to an elliptic fixed point in $\HH$, and otherwise $m=1$.

In what follows we will require various closely related types of Laplacians. In order not to confuse the reader we summarise the various notations.
\begin{notn} 
We distinguish between the following notations.
\begin{enumerate}
\item[(i)] By $\Delta_{z_0}$ we mean the Laplacian $\Delta$ on $D_0$.
\item[(ii)] By $\Delta_{z_0}^*$ we mean the adjoint of $\Delta_{z_0}$ and we denote its domain by $\Dom(\Delta_{z_0}^*)$.
\item[(iii)] We denote by $\overline{\Delta_{z_0}}$ the minimal closed extension of $\Delta_{z_0}$, and we denote its domain by $\overline{D_0}$.
\item[(iv)] By $\Delta_\varphi$, $\varphi\in(-\pi,\pi)$, we mean the one parameter family of self-adjoint extensions of $\Delta_{z_0}$.
\item[(v)] By $\widetilde{\Delta_{z_0}}$ we denote the Laplacian in the distributional sense on $\scrM\setminus\lbrace z_0\rbrace$.
\end{enumerate}
\end{notn}

Denote  
\begin{equation}
H^2(\scrM\setminus\lbrace z_0\rbrace)=\lbrace f\in L^2(\scrM)\mid\widetilde{\Delta_{z_0}}f\in L^2(\scrM)\rbrace.
\end{equation}
With respect to the graph inner product $\left\langle\cdot,\cdot\right\rangle=(\cdot,\cdot)+(\Delta_{z_0}^*\cdot,\Delta_{z_0}^{*}\cdot)$ we have the decomposition (cf. \cite{CdV}, section 1, second paragraph, p. 277)
\begin{equation}
\Dom(\Delta_{z_0}^*)=\overline{D_{0}}\oplus\scrL\lbrace G^{\Gamma}_{\,t}(\cdot,z_{0})\rbrace\oplus\scrL\lbrace G^{\Gamma}_{\,\bar{t}}(\cdot,z_{0})\rbrace
\end{equation}
and it can be seen (cf. \cite{CdV}, Thm. 1, p. 277/8) that $\Dom(\Delta_{z_0}^*)=H^2(\scrM\setminus\lbrace z_0\rbrace)$.
Define the subspace $$\D_{\varphi}\subset\overline{D_{0}}\oplus\scrL\lbrace G^{\Gamma}_{\,t}\left(\cdot,z_{0}\right)\rbrace\oplus\scrL\lbrace G^{\Gamma}_{\,\bar{t}}\left(\cdot,z_{0}\right)\rbrace$$ by
\begin{equation}
D_{\varphi}=\left\{g+cG^{\Gamma}_{\,t}(\cdot,z_{0})+c\,e^{\i\varphi}G^{\Gamma}_{\,\bar{t}}(\cdot,z_{0})|g\in D_{0}, c\in\CC\right\}
\end{equation}
with $\varphi\in(-\pi,\pi)$. This means any $f\in D_{\varphi}$ has, as $z\to z_0$, the asymptotic behaviour
\begin{equation}
f(z)=c(f)\left\{\frac{m}{2\pi}\log d(z,z_{0})+\frac{\gamma}{2\pi}+\Re\psi(t)+k\tan\frac{\varphi}{2}\right\}+o(1),
\end{equation}
where $c(f)\in\CC$ and $$k=\lim_{z\to z_0}\Im G^\Gamma_{t}(z,z_0).$$ Let us introduce the function
\begin{equation}
A(s,t)=\frac{1}{2}\lim_{z\to z_{0}}(G^{\Gamma}_{s}(z,z_{0})-G^{\Gamma}_{t}(z,z_{0}))
\end{equation}
which is meromorphic in $s$. The limit clearly exists, since the logarithmic singularities at $z_{0}$ cancel. The operator $-\Delta_{z_0}$ admits a one parameter family of self-adjoint extensions $\left\{-\Delta_{\varphi}\right\}$, $\varphi\in(-\pi,\pi)$. The operator $-\Delta_{\varphi}: D_{\varphi}\to L^{2}(\Gamma\backslash\HH)$ is defined (cf. \cite{Si}, Section X.1) by
\begin{equation}
-\Delta_{\varphi}f=-\Delta g+c\,t(1-t)G^{\Gamma}_{\,t}(z,z_{0})
+c\,\bar{t}(1-\bar{t})e^{\i\varphi}G^{\Gamma}_{\,\bar{t}}(z,z_{0}).
\end{equation}
It is self-adjoint, see for instance \cite{Si}, Theorem X.2 in Section X.1, p. 140.

It is possible to make an explicit connection between delta potentials and the self-adjoint extensions by relating the coupling constant $\alpha\in\RR$ to the parameter $\varphi\in(-\pi,\pi)$ of the corresponding self-adjoint extension (cf. \cite{Al}, p. 34. Thm. 1.3.2), which amounts to the identity 
\begin{equation}\label{ext}
-2\i\alpha A(t,\bar{t})=\cot\left(\frac{\varphi}{2}\right).
\end{equation}

\section{The spectral function}

In this section we determine a meromorphic function $S_{\alpha,z_{0}}\left(s\right)$ which contains in its zeros and poles all the information about the perturbed and unperturbed discrete spectrum. We begin with an analogue of Hilbert's formula for the iterated resolvent in terms of Green functions.
\begin{lem}\label{Hilbert}
Let $\i=t\left(1-t\right)$, $\lambda=s\left(1-s\right)$. The following identity holds
\begin{equation}
\left(\Delta+\lambda\right)\left(G^{\Gamma}_{s}\left(\cdot,z_{0}\right)-G^{\Gamma}_{t}\left(\cdot,z_{0}\right)\right)
=\left(\i-\lambda\right)G^{\Gamma}_{t}\left(\cdot,z_{0}\right).
\end{equation}
\end{lem}
\begin{proof}
We have the following chain of identities straight from the definition of the Green function,
\begin{equation}
\begin{split}
\left(\Delta+\lambda\right)G^{\Gamma}_{s}(\cdot,z_{0})=\,&\delta_{z_{0}}=(\Delta+\i)G^{\Gamma}_{t}(\cdot,z_{0})\\
=\,&(\Delta+\lambda)G^{\Gamma}_{t}(\cdot,z_{0})+(\i-\lambda)G^{\Gamma}_{t}(\cdot,z_{0}).
\end{split}
\end{equation}
\end{proof}

For $\Re s>1$ we introduce the function
\begin{equation}
S_{\alpha,z_{0}}(s)=\alpha^{-1}+\lim_{z\to z_{0}}(G^{\Gamma}_{s}(z,z_{0})-\tfrac{1}{2}G^{\Gamma}_{t}(z,z_{0})-\tfrac{1}{2}G^{\Gamma}_{\bar{t}}(z,z_{0})),
\end{equation}
which is well defined by the method of images since the sum over the group elements is well known to converge absolutely for $\Re s>1$. 

Next we state some important properties of $S_{\alpha,z_{0}}(s)$. In the previous section we saw (cf. section 2, p. 3) that the discrete spectrum of $-\Delta_\varphi$ consists of two parts: degenerate eigenvalues of the Laplacian (if any) that are inherited by $-\Delta_\varphi$, and new eigenvalues. We refer to the latter part of the discrete spectrum of $-\Delta_\varphi$ as the new part. $S_{\alpha,z_{0}}(s)$ contains information about these eigenvalues in form of its zeros. The following Proposition links the new eigenvalues $\lbrace\lambda_j^\alpha\rbrace$ to the zeros of $S_{\alpha,z_0}$. 
\begin{prop}\label{eigencond}
$\lambda=s(1-s)$ is in the new part of the discrete spectrum of $-\Delta_{\varphi}$ if, and only if, $S_{\alpha,z_{0}}(s)=0$. The corresponding eigenfunctions are given by automorphic Green functions $G^{\Gamma}_{s}(\cdot,z_{0})$. 
\begin{remark}
We remark that the proof makes use only of the self-adjoint extension theory as presented in \cite{Si}. It is possible to give an analogous proof which uses the fact that the self-adjoint extension $-\Delta_\varphi$ is simply the Laplacian in the distributional sense acting on functions with the asymptotic behaviour \eqref{GAsympt}. Hence looking for an eigenfunction of $-\Delta_\varphi$ with eigenvalue $\lambda$ simply means finding a non-trivial element of $\ker(\Delta^*+\lambda)$ with the asymptotic behaviour \eqref{GAsympt}.
\end{remark}
\end{prop}
\begin{proof}
Let $\lambda=s(1-s)$ and $\i=t(1-t)$. Assume that $f_{s}\in\D_{\varphi}\subset L^{2}(\scrM)$ is an eigenfunction of $-\Delta_{\varphi}$ with eigenvalue $\lambda$ and that $\lambda$ does not lie in the discrete spectrum of $\Delta$. By definition
\begin{equation}
(\Delta_{\varphi}+\lambda)f_{s}=0.
\end{equation}
We may write equivalently, using the decomposition of $D_{\varphi}$,
\begin{equation}
(\Delta+\lambda)g_{s}+c(\lambda-\i)G^{\Gamma}_{t}(\cdot,z_{0})+ce^{\i\varphi}(\lambda+\i)G^{\Gamma}_{\bar{t}}(\cdot,z_{0})=0.
\end{equation}
Applying the resolvent on both sides we get
\begin{equation}
g_{s}+c\frac{\lambda-\i}{\Delta+\lambda}G_{t}^{\Gamma}(\cdot,z_{0})+ce^{\i\varphi}\frac{\lambda+\i}{\Delta+\lambda}G_{t}^{\Gamma}(\cdot,z_{0})=0,
\end{equation}
and using Lemma \ref{Hilbert} we rewrite this as
\begin{equation}
g_{s}+c(G^{\Gamma}_{t}(\cdot,z_{0})-G^{\Gamma}_{s}(\cdot,z_{0}))
+ce^{\i\varphi}(G^{\Gamma}_{\bar{t}}(\cdot,z_{0})-G^{\Gamma}_{s}(\cdot,z_{0}))=0.
\end{equation}
We take the limit as $z\to z_{0}$ and obtain
\begin{equation}
cA(s,t)+ce^{\i\varphi}A(s,\bar{t})=0.
\end{equation}
At this point we can divide by $c$ since $c\neq0$. To see this suppose the contrary. It follows that $f_{s}=g_{s}\in D_{0}$. Therefore $0=(\Delta_{\varphi}+\lambda)f_{s}=(\Delta+\lambda)f_{s}$ which contradicts the assumption that $\lambda$ does not lie in the discrete spectrum of $\Delta$. After dividing we have
\begin{equation}\label{altform}
A(s,t)+e^{\i\varphi}A(s,\bar{t})=0,
\end{equation}
which we rewrite as
\begin{equation}\label{speceqn}
\lim_{z\to z_{0}}(G^{\Gamma}_{s}-\tfrac{1}{2}\lbrace G^{\Gamma}_{t}+G^{\Gamma}_{\bar{t}}\rbrace)(z,z_{0}) -\i\tan\frac{\varphi}{2}A(t,\bar{t})=0.
\end{equation}
Finally, using \eqref{ext},
\begin{equation}
\lim_{z\to z_{0}}(G^{\Gamma}_{s}-\tfrac{1}{2}\lbrace G^{\Gamma}_{t}+G^{\Gamma}_{\bar{t}}\rbrace)(z,z_{0})+\alpha^{-1}=0.
\end{equation}

Let us now assume that $S_{\alpha,z_{0}}(s)=0$. We claim that $G^{\Gamma}_{s}(\cdot,z_{0})$ is an eigenfunction of $\Delta_{\varphi}$. We have the decomposition
\begin{equation}
G^{\Gamma}_{s}(z,z_{0})
=\frac{1}{1+e^{\,\i\varphi}}\left\{S_{\,\alpha,\,z_{0}}(z,s)+G^{\Gamma}_{t}(z,z_{0})+e^{\,\i\varphi}G^{\Gamma}_{\bar{t}}(z,z_{0})\right\},
\end{equation}
where we have introduced
\begin{equation}
S_{\alpha,z_{0}}(z,s)=(G^{\Gamma}_{s}-G^{\Gamma}_{t})(z,z_{0})+e^{\,\i\varphi}(G^{\Gamma}_{s}-G^{\Gamma}_{\bar{t}})(z,z_{0}).
\end{equation}
We see from \eqref{altform} and \eqref{speceqn} that $\lim_{z\to z_{0}}S_{\alpha,z_{0}}(z,s)=S_{\alpha,z_{0}}(s)=0$. So $G^{\Gamma}_{s}(\cdot,z_{0})\in D_{\varphi}$. In the above decomposition we multiply through by $1+e^{\,\i\varphi}$ and by definition of $\Delta_{\varphi}$ we obtain
\begin{equation}
\begin{split}
&(1+e^{\,\i\varphi})(\Delta_{\varphi}+\lambda)G^{\Gamma}_{s}(\cdot,z_{0})\\
=&(\Delta+\lambda)S_{\alpha,z_{0}}(\cdot,s)+(\lambda-\i)G^{\Gamma}_{t}(\cdot,z_{0})+e^{\i\varphi}(\lambda+\i)G^{\Gamma}_{t}(\cdot,z_{0}).
\end{split}
\end{equation}
We apply Lemma \ref{Hilbert} to see
\begin{equation}
(\Delta+\lambda)S_{\alpha,z_{0}}(\cdot,s)=(\i-\lambda)G^{\Gamma}_{t}(\cdot,z_{0})+e^{\i\varphi}(-\i-\lambda)G^{\Gamma}_{t}(\cdot,z_{0})
\end{equation}
which implies
\begin{equation}
(1+e^{\,\i\varphi})(\Delta_{\varphi}+\lambda)G^{\Gamma}_{s}(\cdot,z_{0})=0.
\end{equation}
It follows $(\Delta_{\varphi}+\lambda)G^{\Gamma}_{s}(\cdot,z_{0})=0$ since $1+e^{\,\i\varphi}\neq0$. 
\end{proof}

Define
\begin{equation}
\scrI=\lbrace\gamma\in\Gamma\mid\gamma z_{0}=z_{0}\rbrace.
\end{equation}
$\scrI=\lbrace\id\rbrace$ unless $z_{0}$ is an elliptic fixed point in which case $\scrI$ is a finite cyclic group. We can write $S_{\alpha,z_{0}}(s)$ in a more convenient form if $\Re s>1$.
\begin{prop}\label{expression}
Let $\Re s>1$. The function $S_{\alpha,z_{0}}(s)$ can be written in the form
\begin{equation}\label{expr2}
S_{\alpha,z_{0}}(s)=\beta^{-1}+m\psi(s)+\sum_{\gamma\in\Gamma\setminus\scrI}G_{s}(z_{0},\gamma z_{0}),
\end{equation}
where $\beta=\beta(\alpha)$, $\psi(s)=\tfrac{1}{2\pi}\Gamma'(s)/\Gamma(s)$ and $m=|\scrI|$.
\end{prop}
\begin{proof}
It can be seen from the asymptotics \eqref{GAsympt} that
\begin{equation}
\psi(s)-\psi(t)=\lim_{z\to z_{0}}(G_{s}(z,z_{0})-G_{t}(z,z_{0})).
\end{equation}
where $\psi=(2\pi)^{-1}\Gamma'/\Gamma$. From the definition of $S_{\alpha,z_{0}}(s)$ we have for $\Re s>1$
\begin{equation}
\begin{split}
S_{\alpha,z_{0}}(s)&\;=
\alpha^{-1}+\lim_{z\to z_{0}}(G^{\Gamma}_{s}(z,z_{0})-\tfrac{1}{2}G^{\Gamma}_{t}(z,z_{0})-\tfrac{1}{2}G^{\Gamma}_{\bar{t}}(z,z_{0}))\\
&\;=\alpha^{-1}+|\scrI|\psi(s)-|\scrI|\Re\psi(t)\\
&\qquad+\sum_{\gamma\in\Gamma\backslash\scrI}
\lbrace G_{s}(z_{0},\gamma z_{0})-\Re G_{t}(z_{0},\gamma z_{0})\rbrace
\end{split}
\end{equation}
and we let
\begin{equation}\label{const0}
c_{0}=|\scrI|\Re\psi(t)+\Re\sum_{\gamma\in\Gamma\setminus\scrI}G_{t}(z_{0},\gamma z_{0}).
\end{equation}
At this point we choose to reparametrise the coupling constant $\alpha$ (assuming $\alpha\neq-1/c_{0}$) according to
\begin{equation}
\alpha^{-1}-c_{0}=\beta^{-1}
\end{equation}
or 
\begin{equation}
\beta=\frac{\alpha}{1-\alpha c_{0}}.
\end{equation}
We obtain the expression
\begin{equation}\label{geomexpr}
S_{\alpha,z_{0}}(s)=\beta^{-1}+|\scrI|\psi(s)+\sum_{\gamma\in\Gamma\backslash\scrI}G_{s}(z_{0},\gamma z_{0})
\end{equation}
for $\Re s>1$.
\end{proof}

We have a uniform bound on the geometrical terms in the function $S_{\alpha,z_{0}}(s)$ for $\Re s>1$.
\begin{lem}\label{Greenbound}
For all $\rho\in\CC$ with $\Im\rho=-\sigma<-\tfrac{1}{2}$ we have the uniform bound
\begin{equation}\label{sigmabound}
\sum_{\gamma\in\Gamma\backslash\scrI}\left|G_{\tfrac{1}{2}+\i\rho}\left(z_{0},\gamma z_{0}\right)\right|\ll_{\Gamma,z_{0}}
\footnote{
The notation $f\ll_\bullet g$ means that there is a constant depending on $\bullet$ s.t $f\leq C(\bullet)g$.
}
\sigma^{-1/2}
\end{equation}
\end{lem}
\begin{proof}
Let $\tau_{0}=\inf\lbrace d(\gamma z_{0},z_{0})\mid\gamma\in\Gamma\backslash\scrI\rbrace$. Discreteness of $\Gamma$ implies $\tau_{0}>0$. We will make use of the integral representation of the free Green function
\begin{equation}\label{intrep}
G_{\tfrac{1}{2}+\i\rho}(z,w)=-\frac{1}{2\pi\sqrt{2}}\int_{d(z,w)}^{\infty}\frac{e^{-\i\rho t}dt}{\sqrt{\cosh t-\cosh d(z,w)}}
\end{equation}
which is valid for $\Im\rho<-\tfrac{1}{2}$. Let $\tau_{\gamma}=d(z_{0},\gamma z_{0})$. We have
\begin{equation}\label{exactbound}
\begin{split}
&\sum_{\gamma\in\Gamma\backslash\scrI}\left|G_{\tfrac{1}{2}+\i\rho}(z_{0},\gamma z_{0})\right|\\
&\leq\frac{1}{2\pi\sqrt{2}}\sum_{\gamma\in\Gamma\backslash\scrI}\int_{\tau_{\gamma}}^{\infty}\frac{e^{-\sigma t}dt}{\sqrt{\cosh t-\cosh\tau_{\gamma}}}\\
&=\frac{1}{2\pi\sqrt{2}}\sum_{\gamma\in\Gamma\backslash\scrI}e^{-\sigma\tau_{\gamma}}\int_{0}^{\infty}\frac{e^{-\sigma t}dt}{\sqrt{\cosh (t+\tau_{\gamma})-\cosh\tau_{\gamma}}}\\
&\leq\,\frac{C_{\epsilon}}{2\pi\sqrt{2\sinh\tau_{0}}}\int_{0}^{\infty}\frac{e^{-\sigma t}dt}{\sqrt{t}}\\
&=\frac{C_{\epsilon}}{\sqrt{4\pi\sigma\sinh\tau_{0}}}=C(\Gamma,z_{0})\,\sigma^{-1/2}
\end{split}
\end{equation}
since for $t>0$ and any $\gamma\in\Gamma\backslash\scrI$
\begin{equation}
\sinh\tau_{0}\leq\sinh\tau_{\gamma}\leq\frac{\cosh(\tau_{\gamma}+t)-\cosh\tau_{\gamma}}{t}
\end{equation}
and where $-\sigma<\epsilon<-1$ such that
\begin{equation}
\sum_{\gamma\in\Gamma\backslash\scrI}e^{-\sigma\tau_{\gamma}}\leq\sum_{\gamma\in\Gamma\backslash\scrI}e^{\epsilon\tau_{\gamma}}
\leq C_{\epsilon}
\end{equation}
where for $r<\frac{1}{2}\tau_{0}$ (cf. Lemma 5 in \cite{Mf}, p. 19)
\begin{equation}
C_{\epsilon}=\frac{2\pi e^{-2\pi\epsilon r}}{\Area(r)}\int_{0}^{\infty}e^{\epsilon\tau}\sinh\tau d\tau.
\end{equation}
\end{proof}

The function $S_{\alpha,z_{0}}(s)$ contains all information about the unperturbed and perturbed discrete spectrum, as well as unperturbed and perturbed resonances in form of its poles and zeros. The following Theorem locates those and gives their spectral interpretation.
\begin{thm}\label{prop}
$S_{\alpha,z_{0}}(s)$ has the following zeros and poles.\\
\begin{itemize}
\item[(i)] There are simple poles at $\tfrac{1}{2}+\i r_{j}$ and $\tfrac{1}{2}-\i r_{j}$ corresponding to eigenvalues $\mu_j=\tfrac{1}{4}+r_{j}^{2}$, $r_{j}\in\RR\cup\i\RR$.
\item[(ii)] There are simple zeros at $\tfrac{1}{2}+\i r^{\alpha}_{j}$ and $\tfrac{1}{2}-\i r^{\alpha}_{j}$, located in between the poles above, corresponding to new eigenvalues $\mu_j^\alpha=\tfrac{1}{4}+(r^{\alpha}_{j})^{2}$, $r^{\alpha}_{j}\in\RR\cup\i\RR$.\\
\end{itemize}
\end{thm}
\begin{proof}
We have from \eqref{greenspec}, for $t(1-t)=\tfrac{1}{4}+\xi^{2}$ and $\mu_j=\tfrac{1}{4}+r_j^2$ with multiplicity $m_j$,
\begin{equation}
\begin{split}
S_{\alpha,z_{0}}(\tfrac{1}{2}+\i\rho)=&\alpha^{-1}+\sum_{j=0}^{\infty}|\varphi_{j}(z_{0})|^{2}\left\{\frac{1}{\rho_{j}^{2}-\rho^{2}}-\Re\left\{\frac{1}{\rho_{j}^{2}-\xi^{2}}\right\}\right\}\\
=&\alpha^{-1}+\sum_{j=0}^{\infty}m_j|\varphi_{j}(z_{0})|^{2}\left\{\frac{1}{r_{j}^{2}-\rho^{2}}-\Re\left\{\frac{1}{r_{j}^{2}-\xi^{2}}\right\}\right\}
\end{split}
\end{equation}
To see that the sum on the right is absolutely convergent, let $\lambda_j=\tfrac{1}{4}+\rho_j^2$, $\lambda=\tfrac{1}{4}+\rho^2$ and rewrite 
\begin{equation}
\begin{split}
&\sum_{j=0}^{\infty}|\varphi_{j}(z_{0})|^{2}\left\{\frac{1}{\lambda_j-\lambda}-\Re\left\{\frac{1}{\lambda_j-\i}\right\}\right\}\\
=&\sum_{j=0}^{\infty}|\varphi_{j}(z_{0})|^{2}\left\{\frac{1}{\lambda_j-\lambda}-\frac{\lambda_j}{\lambda_j^2+1}\right\}\\
=&\sum_{j=0}^{\infty}|\varphi_{j}(z_{0})|^{2}\frac{1+\lambda_j\lambda}{(\lambda_j-\lambda)(\lambda_j^2+1)}
\end{split}
\end{equation}
and absolute convergence follows from Weyl's law and the standard bound $|\varphi_j(z_0)|\ll\lambda_j^{1/4}$.

For $\rho\in\RR\cup\i\RR$, and depending on whether $\rho=v$ or $\rho=\i v$, for $v\in\RR$,
\begin{equation}
\frac{d}{dv}S_{\alpha,z_{0}}(\tfrac{1}{2}+\i\rho(v))=\pm v\sum_{j=0}^{\infty}\frac{|\varphi_{j}(z_{0})|^{2}}{(\rho_{j}^{2}\pm v^{2})^{2}}
\end{equation}
which shows that the zeros of $S_{\alpha,z_{0}}(s)$ lie in between the poles on the critical line and the real line.
\end{proof}

\section{The trace formula}

In this section we will give the proof of Theorem \ref{thm1}. We first prove a truncated trace formula. Recall
\begin{equation}
S_{\alpha,z_{0}}(s)=\alpha^{-1}+\lim_{z\to z_{0}}\lbrace G_{s}^{\Gamma}(z,z_{0})-\Re G_{t}^{\Gamma}(z,z_{0})\rbrace.
\end{equation}

\begin{prop}
Let $h\in H_{\sigma,\delta}$ and $T>0$. Define 
\begin{equation}
B(T)=\lbrace\rho\in\CC\mid|\Im\rho|<\sigma,\;|\Re\rho|<T\rbrace.
\end{equation} 
Then 
\begin{equation}\label{trunc}
\begin{split}
\sum_{\rho_{j}^{\alpha}\in B(T)}h(\rho^{\alpha}_{j})-\sum_{\rho_{j}\in B(T)}h(\rho_{j})=&\frac{1}{\pi\i}\int_{-\i\sigma-T}^{-\i\sigma+T}h(\rho)\frac{S'_{\alpha,z_{0}}}{S_{\alpha,z_{0}}}(\tfrac{1}{2}+\i\rho)d\rho\\
&+\frac{1}{\pi\i}\int_{-\i\sigma+T}^{\i\sigma+T}h(\rho)\frac{S'_{\alpha,z_{0}}}{S_{\alpha,z_{0}}}(\tfrac{1}{2}+\i\rho)d\rho.
\end{split}
\end{equation}
\end{prop}
\begin{proof}
Let $t=\tfrac{1}{2}+\i\xi$. Since $\scrM$ is compact, we have for an orthonormal basis of eigenfunctions of the Laplacian $\lbrace\varphi_{j}\rbrace_{j=0}^{\infty}$ the spectral expansion (which we recall is absolutely convergent)
\begin{equation}
S_{\alpha,z_{0}}(\tfrac{1}{2}+\i\rho)=\alpha^{-1}+\sum_{j=0}^{\infty}|\varphi_{j}(z_{0})|^{2}\left\{\frac{1}{\rho_{j}^{2}-\rho^{2}}-\Re\left\{\frac{1}{\rho_{j}^{2}-\xi^{2}}\right\}\right\}
\end{equation}
where we may rewrite (let $m_j$ be the multiplicity of $\lambda_j$) the sum on the right as
$$\sum_{j=0}^{\infty}m_j|\varphi_{j}(z_{0})|^{2}\left\{\frac{1}{r_{j}^{2}-\rho^{2}}-\Re\left\{\frac{1}{r_{j}^{2}-\xi^{2}}\right\}\right\}.$$
We obtain after a contour integration along $\partial B(T)$
\begin{equation}
\begin{split}
\sum_{\rho_{j}^{\alpha}\in B(T)}h(r^{\alpha}_{j})-\sum_{r_{j}\in B(T)}h(r_{j})
=&\frac{1}{2\pi\i}\int_{\partial B(T)}h(\rho)\frac{S'_{\alpha,z_{0}}}{S_{\alpha,z_{0}}}(\tfrac{1}{2}+\i\rho)d\rho\\
=&\frac{1}{\pi\i}\int_{-\i\sigma-T}^{-\i\sigma+T}h(\rho)\frac{S'_{\alpha,z_{0}}}{S_{\alpha,z_{0}}}(\tfrac{1}{2}+\i\rho)d\rho\\
&+\frac{1}{\pi\i}\int_{-\i\sigma+T}^{\i\sigma+T}h(\rho)\frac{S'_{\alpha,z_{0}}}{S_{\alpha,z_{0}}}(\tfrac{1}{2}+\i\rho)d\rho
\end{split}
\end{equation}
and we recall the identity \eqref{specid} to see the result.
\end{proof}

In order to prove the full trace formula we must show
\begin{equation}\label{bdtermstozero}
\lim_{T\to\infty}\int_{-\i\sigma+T}^{\i\sigma+T}h(\rho)\frac{S'_{\alpha,z_{0}}}{S_{\alpha,z_{0}}}(\tfrac{1}{2}+\i\rho)d\rho=0
\end{equation}
for $T$ such that $\partial B(T)$ does not contain any zeros or poles of $S_{\alpha,z_{0}}$. We can conclude from equation \eqref{trunc} that the limit of the boundary terms \eqref{bdtermstozero} exists. This follows from the absolute convergence of the perturbed and unperturbed traces $$\sum_{\rho_{j}^{\alpha}\in B(T)}h(\rho^{\alpha}_{j}),\;\sum_{\rho_{j}\in B(T)}h(\rho_{j})$$ as well as the expression \eqref{expr2} in Proposition \ref{expression} and the bound \eqref{sigmabound} on $S_{\alpha,z_{0}}(\tfrac{1}{2}+\i\rho)$ in Lemma \ref{Greenbound} which holds for $\sigma>\tfrac{1}{2}$ along the line  $\Im\rho=-\sigma$, i.e in particular for the line segment $[-T-\i\sigma,T-\i\sigma]$.

To see that the integral $$\frac{1}{2\pi\i}\int_{-\i\sigma-T}^{-\i\sigma+T}h(\rho)\frac{S'_{\alpha,z_{0}}}{S_{\alpha,z_{0}}}(\tfrac{1}{2}+\i\rho)d\rho$$
converges, observe 
\begin{equation}
\begin{split}
\int_{-\i\sigma-T}^{-\i\sigma+T}h(\rho)\frac{S'_{\alpha,\,z_{0}}}{S_{\alpha,\,z_{0}}}(\tfrac{1}{2}+\i\rho)d\rho
=\,&h(-\i\sigma+T)\log S_{\alpha,z_{0}}(\tfrac{1}{2}+\sigma+\i T)\\
&-h(-\i\sigma-T)\log S_{\alpha,z_{0}}(\tfrac{1}{2}+\sigma-\i T)\\
&-\int_{-\i\sigma-T}^{-\i\sigma+T}h'(\rho)\log S_{\alpha,\,z_{0}}(\tfrac{1}{2}+\i\rho)d\rho
\end{split}
\end{equation}
and, because of the asymptotics as $T\to\infty$
\begin{equation}\label{SAsymp}
|S_{\alpha,\,z_{0}}(\tfrac{1}{2}+\sigma\pm\i T)|\sim |\psi(\tfrac{1}{2}+\sigma\pm\i T)| \sim \log T,
\end{equation}
the fact that the winding number of $S_{\alpha,z_{0}}(\tfrac{1}{2}+\sigma+\i r)$ is constant for $r\in\RR$ and the decay of $h$ we have $$\lim_{T\to\infty}h(-\i\sigma\pm T)\log S_{\alpha,z_{0}}(\tfrac{1}{2}+\sigma\pm\i T)=0.$$

A straightforward application of Cauchy's Theorem yields $|h'(\rho)|\ll (1+|\Re\rho|)^{-2-\delta}$ for $|\Im\rho|<\sigma$. Note that the contour of integration does not contain any poles or zeros of $S_{\alpha,z_0}$ and the only poles or zeros off the real line can be found on the line segment $(-\i\sigma,\i\sigma)$. So we may shift the contour of integration slightly and in view of the decay of $h'$ and the asymptotics \eqref{SAsymp} we see  that the integral $$\int_{-\i\sigma-T}^{-\i\sigma+T}h'(\rho)\log S_{\alpha,\,z_{0}}(\tfrac{1}{2}+\i\rho)d\rho$$ converges.

In order to show that the limit in \eqref{bdtermstozero} is zero we construct a subsequence of line segments along which the integral tends to zero. The vanishing of the boundary terms \eqref{bdtermstozero} is a simple application of the following Proposition (see Theorem \ref{van}). We give a bound on the integral of $|\log |S_{\alpha,z_{0}}||$ on a sequence of line segments $[t_l-\i\sigma,t_l]$.
\begin{prop}\label{compbound}
There exists an increasing sequence $\lbrace t_{l}\rbrace_{l=0}^{\infty}\subset\RR_+$, such that $t_{l}\stackrel{l\to\infty}{\longrightarrow}\infty$, and for any $\epsilon>0$
\begin{equation}\label{testbound}
\int_{t_{l}-\i\sigma}^{t_{l}}\big|\log|S_{\alpha,z_{0}}(\tfrac{1}{2}+\i\rho)|\big||d\rho|\ll_\epsilon t_{l}^{2+\epsilon}.
\end{equation}
\end{prop}
The strategy in proving Proposition \ref{compbound} is to construct (see Lemma \ref{testf}) a suitably symmetric test function $h_\epsilon\in H_{\epsilon,\sigma}$ s.t for all $\rho\in[t_l-\i\sigma,t_l]$ and $t_l$ sufficiently large we have $|\Re h'_\epsilon(\rho)|=\Re h'_\epsilon(\rho)\gg t_l^{-2-\epsilon}$. The existence of the limit for each $\epsilon>0$ and $h_\epsilon$ as above will enable us to derive the bound \eqref{testbound}.
We also require a bound on the spectral function on the line segments $[t_l-\i\sigma,t_l]$. The reason for this is that if we know $|S_{\alpha,z_0}|\leq c t_l^A$, $A\in\NN$, on these line segments, we have the bound 
\begin{equation}\label{logbound}
\begin{split}
|\log|S_{\alpha,z_0}||
&\leq |\log (c^{-1}t_l^{-A}|S_{\alpha,z_0}|)|+|\log c|+|A||\log t_l|\\
&= -\log (c^{-1}t_l^{-A}|S_{\alpha,z_0}|)+|\log c|+|A||\log t_l|
\end{split}
\end{equation} 
where we can drop the absolute value for the first term, since $\log (c^{-1}t_l^{-A}|S_{\alpha,z_0}|)\leq0$ in view of the bound $|S_{\alpha,z_0}|\leq c t_l^A$. This will be essential in the proof of Proposition \ref{compbound}.

We first construct a sequence $\lbrace T_{N}\rbrace_{N}\subset\RR_+$ s.t $S_{\alpha,z_{0}}(\tfrac{1}{2}+\i\rho)$ admits a uniform polynomial bound on the corresponding sequence of line segments. We then select a subsequence $\lbrace T_{N(l)}\rbrace_{l}\subset\RR_+$ s.t $|\Re h'_\epsilon(\rho)|=\Re h'_\epsilon(\rho)\gg T_{N(l)}^{-2-\epsilon}$ for sufficiently large $T_{N(l)}$.
\begin{prop}\label{polybound}
It exists a sequence $\lbrace T_{N}\rbrace_{N}$ in $\RR_{+}$, $\lim_{N} T_{N}=+\infty$ s.t uniformly $\forall N,\forall w\in[-\sigma,0]$
\begin{equation}\label{polyn}
\sum_{j=0}^{\infty}|\varphi_{j}(z_{0})|^{2}\left|\frac{1}{\lambda_{j}-\mu_{N}(w)}-\frac{1}{\lambda_{j}-\i}\right|\ll_{\Gamma}T_{N}^{5}
\end{equation}
where $\mu_{N}(w)=\tfrac{1}{4}+(T_{N}+\i w)^{2}$.
\end{prop}

Before we give the proof of Proposition \ref{polybound} we state a Lemma which will play a central role in the proof. 
\begin{lem}\label{spacing}
Let $\lbrace\varphi_{j}\rbrace_{j}$ be the set of eigenfunctions on $\Gamma\backslash\HH$ with $(\Delta+\lambda_{j})\varphi_{j}=0$. Then there exists a subsequence $\lbrace\lambda_{j_{k}}\rbrace_{k=0}^{\infty}\subset\lbrace\lambda_{j}\rbrace_{j=0}^{\infty}$ and a constant $c_{0}(\Gamma)>0$ s. t. $\lambda_{j_{k}+1}-\lambda_{j_{k}}\geq c_{0}(\Gamma)$.
\end{lem}
\begin{proof}
Weyl's law (cf. \cite{Bd})
\begin{equation}
\#\lbrace j\mid\lambda_{j}\leq T\rbrace=\frac{\Area(\scrM)}{4\pi}T+O(T^{1/2}/\log T)
\end{equation}
implies that there exist constants $c_{2}>c_{1}>0$ such that for any integer $n$
\begin{equation}
c_{1}n\leq\lambda_{n}\leq c_{2}n.
\end{equation}
Choose any integers $m$, $n$ such that $n-m$ is even and $n/m>2c_{2}/c_{1}-1$. Then
\begin{equation}
\lambda_{n}-\lambda_{m}=(\lambda_{n}-\lambda_{n-1})+\cdots+(\lambda_{m+1}-\lambda_{m})\geq c_{1}n-c_{2}m
>\tfrac{1}{2}c_{1}(n-m).
\end{equation}
Since the number of terms above is $n-m$ and all terms are nonnegative, it follows that at least one of them is $\geq\tfrac{1}{2}c_{1}$. So in every interval $[\lambda_{m},\lambda_{n}]$ we find a pair $(\lambda_{k},\lambda_{k+1})$ with $\lambda_{k+1}-\lambda_{k}\geq\tfrac{1}{2}c_{1}$. By choosing an appropriate sequence of intervals we construct the sequence $\lbrace\lambda_{j_{k}}\rbrace_{k=0}^{\infty}$.
\end{proof}

We apply Lemma \ref{spacing} to prove the Proposition.\\

\textit{Proof of Proposition \ref{polybound}.}
By Lemma \ref{spacing} we can choose an infinite increasing subsequence of Laplacian eigenvalues $\lbrace \lambda_{k(N)}\rbrace_N$ such that $\lambda_{k(N)+1}-\lambda_{k(N)}=\rho_{k(N)+1}^2-\rho_{k(N)}^2\gg1$. Recall that between two distinct consecutive eigenvalues $\lambda_{k(N)}=\tfrac{1}{4}+\rho_{k(N)}^2$ and $\lambda_{k(N)+1}=\tfrac{1}{4}+\rho_{k(N)+1}^2$ there is exactly one new eigenvalue $\mu^\alpha_{k(N)}=\tfrac{1}{4}+\chi_{k(N)}^2$ and $\chi_{k(N)}\in(\rho_{k(N)},\rho_{k(N)+1})\subset\RR_{+}$ is a zero of the function $S_{\alpha,z_0}(\tfrac{1}{2}+\i\rho)$, whereas $\rho_{k(N)}$, $\rho_{k(N)+1}$ are singularities of the same function.

So we may choose an infinite sequence
\begin{equation}\label{infseq}
T_{N}=
\begin{cases}
\tfrac{1}{2}(\rho_{k(N)}+\chi_{k(N)}),\;\text{if}\;|\chi_{k(N)}-\rho_{k(N)}|\geq|\chi_{k(N)}-\rho_{k(N)+1}|\\
\tfrac{1}{2}(\rho_{k(N)+1}+\chi_{k(N)}),\;\text{otherwise.}
\end{cases}
\end{equation}
with $|\rho_{k(N)}-\rho_{k(N)+1}|\gg|\rho_{k(N)}+\rho_{k(N)+1}|\asymp T_{N}^{-1}$. Note in particular that for all $\rho_j\in\RR_+$, 
\begin{equation}\label{lboundTN}
|\rho_j-T_N|\geq\tfrac{1}{4}|\rho_{k(N)}-\rho_{k(N)+1}|\gg T_{N}^{-1}.
\end{equation}

Let $\mu_{N}(w)=\tfrac{1}{4}+(T_{N}+\i w)^{2}$, $w\in[-\sigma,0]$. We have
\begin{equation}
\begin{split}
&\sum_{j=0}^{\infty}|\varphi_{j}(z_{0})|^{2}\left|\frac{1}{\lambda_{j}-\mu_{N}(w)}-\frac{1}{\lambda_{j}-\i}\right|\\
&\ll |\i-\mu_{N}(w)|\sum_{j=0}^{\infty}\frac{\lambda_{j}^{1/2}}{|\lambda_{j}-\mu_{N}(w)||\lambda_{j}-\i|}
\end{split}
\end{equation}
where we have used the bound $|\varphi_{j}(z_{0})|^{2}\ll\lambda_{j}^{1/2}$ (cf. \cite{Iw}, p. 108, (8.3')). Fix $\beta\in(\tfrac{1}{2},1)$. We split the sum into a central part satisfying $\inf_{w\in[-\sigma,0]}|\lambda_{j}-\mu_{N}(w)|<\lambda_{j}^{\beta}$ and a corresponding tail. For convenience we let $I_{N}(\lambda_{j})=\inf_{w\in[-\sigma,0]}|\lambda_{j}-\mu_{N}(w)|$. The first sum is estimated by
\begin{equation}
\begin{split}
&\sum_{I_{N}(\lambda_{j})<\lambda_{j}^{\beta}}\frac{\lambda_{j}^{1/2}}{|\lambda_{j}-\mu_{N}(w)||\lambda_{j}-\i|}\\
\leq&\#\lbrace j\mid I_{N}(\lambda_{j})<\lambda_{j}^{\beta}\rbrace\, \max_{I_{N}(\lambda_{j})<\lambda_{j}^{\beta}}\,\sup_{w\in[-\sigma,0]}\left\{\frac{\lambda_{j}^{1/2}}{|\lambda_{j}-\mu_{N}(w)||\lambda_{j}-\i|}\right\}.
\end{split}
\end{equation}
Now if $\lambda_{j}>\tfrac{1}{4}+T_{N}^{2}$ then $I_{N}(\lambda_{j})=\lambda_{j}-\tfrac{1}{4}-T_{N}^{2}$. It follows
\begin{equation}
\begin{split}
&\#\lbrace j\mid I_{N}(\lambda_{j})<\lambda_{j}^{\beta}\rbrace\\
\leq&\#\lbrace j\mid\lambda_{j}\leq\tfrac{1}{4}+T_{N}^{2}\rbrace
+\#\lbrace j\mid \lambda_{j}-\lambda_{j}^{\beta}<\tfrac{1}{4}+T_{N}^{2}\rbrace.
\end{split}
\end{equation}
Let
\begin{equation}
C(\beta)=\#\lbrace j\mid\lambda_{j}\leq 2^{1/(1-\beta)}\rbrace
\end{equation}
and observe that $\lambda_{j}>2^{1/(1-\beta)}$ implies $\lambda_{j}^{\beta-1}<\tfrac{1}{2}$. So $\lambda_{j}>2^{1/(1-\beta)}$ together with $\lambda_{j}(1-\lambda_{j}^{\beta-1})<\tfrac{1}{4}+T_{N}^{2}$ implies
\begin{equation}
\lambda_{j}<2\lambda_{j}(1-\lambda_{j}^{\beta-1})<\tfrac{1}{2}+2T_{N}^{2}.
\end{equation}
Hence
\begin{equation}
\begin{split}
&\#\lbrace j\mid\lambda_{j}(1-\lambda_{j}^{\beta-1})<\tfrac{1}{4}+T_{N}^{2}\rbrace\\
\leq&\,\#\lbrace j\mid\lambda_{j}\leq2^{1/(1-\beta)},\;\lambda_{j}(1-\lambda_{j}^{\beta-1})<\tfrac{1}{4}+T_{N}^{2}\rbrace\\
&+\#\lbrace j\mid\lambda_{j}>2^{1/(1-\beta)},\;\lambda_{j}(1-\lambda_{j}^{\beta-1})<\tfrac{1}{4}+T_{N}^{2}\rbrace\\
\leq&\,C(\beta)+\#\lbrace j\mid2^{1/(1-\beta)}<\lambda_{j}<\tfrac{1}{2}+2T_{N}^{2}\rbrace\\
\ll&\; T_N^2.
\end{split}
\end{equation}
It follows that
\begin{equation}
\#\lbrace j\mid I_{N}(\lambda_{j})<\lambda_{j}^{\beta}\rbrace\ll_{\beta}T_{N}^{2}.
\end{equation}
By the same observations as above we see that $I(\lambda_{j})<\lambda_{j}^{\beta}$ implies $\lambda_{j}\leq\max\lbrace2^{1/(1-\beta)},\tfrac{1}{2}+2T_{N}^{2}\rbrace$. Also for any $j\geq0$ we have (see \eqref{lboundTN}) $$|\rho_{j}-T_{N}|\geq\tfrac{1}{4}|\rho_{k(N)}-\rho_{k(N)+1}|\gg T_{N}^{-1}$$ which implies
\begin{equation}
\begin{split}
|\lambda_{j}-\mu_{N}(w)|=|\rho_{j}^{2}-(T_{N}+\i w)^{2}|
=&\;|\rho_{j}-T_{N}-\i w||\rho_{j}+T_{N}+\i w|\\
\geq&\;|\rho_{j}-T_{N}|(\rho_{j}+T_{N})\\
\gg&\;1.
\end{split}
\end{equation}
Since $|\lambda_{j}-\i|\geq1$ we have
\begin{equation}
\max_{I_{N}(\lambda_{j})<\lambda_{j}^{\beta}}\,\sup_{w\in[-\sigma,0]}\left\{\frac{\lambda_{j}^{1/2}}{|\lambda_{j}-\mu_{N}(w)||\lambda_{j}-\i|}\right\}
\ll T_{N}.
\end{equation}

The tail can be bounded as follows
\begin{equation}
\begin{split}
\sum_{I(\lambda_{j})\geq\lambda_{j}^{\beta}}\frac{\lambda_{j}^{1/2}}{|\lambda_{j}-\mu_{N}(w)||\lambda_{j}-\i|}
\leq&\sum_{I(\lambda_{j})\geq\lambda_{j}^{\beta}}\frac{\lambda_{j}^{1/2-\beta}}{|\lambda_{j}-\i|}\\
\leq&\sum_{j=0}^{\infty}\frac{\lambda_{j}^{1/2-\beta}}{|\lambda_{j}-\i|}<+\infty.
\end{split}
\end{equation}
Finally note that $|\mu_{N}(w)-\i|\ll T_{N}^{2}$.
\begin{flushright}
$\square$
\end{flushright}

The following Lemma establishes the existence of a test function in the space $H_{\sigma,\epsilon/2}$ for any $\epsilon>0$ with certain properties which we will use in the proof of Proposition \ref{compbound}. The construction is technical and we provide it in appendix A.
\begin{lem}\label{testf}
Let $\epsilon>0$ and recall the sequence $\lbrace T_N\rbrace_N$ defined by \eqref{infseq}. There exists $h_{\epsilon}\in H_{\sigma,\epsilon/2}$ such that
\begin{equation}\label{sym2}
\overline{h_{\epsilon}(\rho)}=h_{\epsilon}(\overline{\rho})
\end{equation}
and for some subsequence $\lbrace T_{N(l)}\rbrace_{l}\subset\lbrace T_{N}\rbrace_{N}$, $\lim_{l\to\infty}T_{N(l)}=\infty$ and $\rho\in[T_{N(l)},T_{N(l)}-\i\sigma]$ and $T_{N(l)}$ sufficiently large
\begin{equation}
|\Re h'_{\epsilon}(\rho)|=\Re h'_{\epsilon}(\rho)\gg T_{N(l)}^{-2-\epsilon}
\end{equation}
uniformly in $\rho$.
\end{lem}

We can now use Lemma \ref{testf} to prove the Proposition.\\

\textit{Proof of Proposition \ref{compbound}.} 
There exists (cf. Proposition \ref{polybound}) an increasing sequence $\lbrace T_{N}\rbrace_{N=0}^{\infty}$, $\lim_N T_N=\infty$, such that $|S_{\alpha,z_{0}}(\tfrac{1}{2}+\i\rho)|< c(\Gamma)T_{N}^{5}$ for $\rho\in[T_{N},T_{N}-\i\sigma]$. We recall (see the argument on pp. 14-15 ) that
\begin{equation}\label{bdterms}
\lim_{N\to\infty}\int_{-\i\sigma+T_{N}}^{\i\sigma+T_{N}}h(\rho)\frac{S'_{\alpha,z_{0}}}{S_{\alpha,z_{0}}}(\tfrac{1}{2}+\i\rho)d\rho
\end{equation}
exists. We have $\arg(S_{\alpha,z_{0}}(\tfrac{1}{2}+\i\rho))\ll1$ along the edges of $B(T)$ since the winding number about the poles is less than $1$. Thus, for $N\to\infty$,
\begin{equation}
|\log S_{\alpha,z_{0}}(\tfrac{1}{2}\pm\sigma+\i T_{N})|
\ll\log|S_{\alpha,z_{0}}(\tfrac{1}{2}\pm\sigma+\i T_{N})|\sim\log\log T_{N}
\end{equation}
which follows (using evenness in $\rho$) from
\begin{equation}
\begin{split}
|S_{\alpha,z_{0}}(\tfrac{1}{2}-\sigma+\i T_{N})|
=&|S_{\alpha,z_{0}}(\tfrac{1}{2}+\sigma-\i T_{N})|\\
=&\left|\overline{S_{\alpha,z_{0}}(\tfrac{1}{2}+\sigma-\i T_{N})}\right|\\
=&|S_{\alpha,z_{0}}(\tfrac{1}{2}+\sigma+\i T_{N})|\asymp \log T_{N}
\end{split}
\end{equation} 
and the last line follows in view of equation \eqref{geomexpr}, Lemma \ref{Greenbound} as well as the approximation (cf. \cite{Iw}, p. 199, eq. (B11)) $$\psi(s)=\frac{1}{2\pi}\log s-\frac{1}{4\pi s}+O(|s|^{-1})$$ which holds uniformly for $|\arg s|<\pi-\epsilon'$ for some small $\epsilon'>0$.

Since $h(\pm\i\sigma+T_{N})\ll T_{N}^{-2-\delta}$, we conclude by integration by parts that
\begin{equation}\label{limitbdterms}
\lim_{N\to\infty}\int_{-\i\sigma+T_{N}}^{\i\sigma+T_{N}}h'(\rho)\log|S_{\alpha,z_{0}}(\tfrac{1}{2}+\i\rho)|d\rho
\end{equation}
exists.

Now we know (cf. Lemma \ref{testf}) that there is a subsequence $\lbrace T_{N(l)}\rbrace_{l}\subset\lbrace T_{N}\rbrace_{N}$ and $h_{\epsilon}\in H_{\sigma,\epsilon/2}$ (for any $\epsilon>0$) such that
\begin{equation}\label{sym}
h'_{\epsilon}(\bar{\rho})=-h'_{\epsilon}(\rho)
\end{equation}
and for $\rho\in[T_{N(l)},T_{N(l)}-\i\sigma]$ and $T_{N(l)}$ sufficiently large we have
\begin{equation}
|\Re h'_{\epsilon}(\rho)|=\Re h'_{\epsilon}(\rho)\gg T_{N(l)}^{-2-\epsilon}.
\end{equation}
We have, since the integrand is odd and because of \eqref{sym},
\begin{equation}\label{id}
\begin{split}
&\int_{-\i\sigma+T_{N(l)}}^{\i\sigma+T_{N(l)}}h'_{\epsilon}(\rho)\log|S_{\alpha,z_{0}}(\tfrac{1}{2}+\i\rho)|d\rho\\
=&\left\{\int_{-T_{N(l)}}^{-\i\sigma-T_{N(l)}}+\int_{-\i\sigma+T_{N(l)}}^{T_{N(l)}}\right\}h'_{\epsilon}(\rho)\log|S_{\alpha,z_{0}}(\tfrac{1}{2}+\i\rho)|d\rho\\
=&-2\i\int_{-\i\sigma+T_{N(l)}}^{T_{N(l)}}\Re h'_{\epsilon}(\rho)\log|S_{\alpha,z_{0}}(\tfrac{1}{2}+\i\rho)|d\rho.
\end{split}
\end{equation}
Since $|\Re h'_{\epsilon}(\rho)|\ll(1+|\Re\rho|)^{-2-\epsilon/2}$, we conclude in view of \eqref{id} and \eqref{limitbdterms} that
\begin{equation}\label{limexistsalso}
\lim_{l\to\infty}\int_{-\i\sigma+T_{N(l)}}^{T_{N(l)}}\Re h'_{\epsilon}(\rho)\log\left\{ c(\Gamma)^{-1}T_{N(l)}^{-5}|S_{\alpha,z_{0}}(\tfrac{1}{2}+\i\rho)|\right\}d\rho
\end{equation}
exists. And since for all $\rho\in[T_{N(l)},T_{N(l)}-\i\sigma]$ and $T_{N(l)}$ sufficiently large
\begin{equation}
\Re h'_{\epsilon}(\rho)\log\left\{c(\Gamma)^{-1}T_{N(l)}^{-5}|S_{\alpha,z_{0}}(\tfrac{1}{2}+\i\rho)|\right\}<0
\end{equation}
we have
\begin{equation}\label{last}
\begin{split}
&T_{N(l)}^{-2-\epsilon}\int_{0}^{\sigma}\left|\log\left\{c(\Gamma)^{-1}T_{N(l)}^{-5}|S_{\alpha,z_{0}}(\tfrac{1}{2}+\sigma-w+\i T_{N(l)})|\right\}\right|dw\\
\ll&\int_{0}^{\sigma}\Re h'_{\epsilon}(\i(w-\sigma)+T_{N(l)})\\
&\times\left|\log \left\{c(\Gamma)^{-1}T_{N(l)}^{-5}|S_{\alpha,z_{0}}(\tfrac{1}{2}+\sigma-w+\i T_{N(l)})|\right\}\right|dw\\
\ll&\left|\int_{0}^{\sigma}\Re h'_{\epsilon}(\i(w-\sigma)+T_{N(l)})\log|S_{\alpha,z_{0}}(\tfrac{1}{2}+\sigma-w+\i T_{N(l)})|dw\right|\\
&\;+O(T_{N(l)}^{-2-\epsilon}\log T_{N(l)}).
\end{split}
\end{equation}
The first term on the RHS of \eqref{last} converges as $l\to\infty$ (cf. \eqref{limexistsalso}), which implies 
\begin{equation}
\int_{0}^{\sigma}|\log\lbrace c(\Gamma)^{-1}T_{N(l)}^{-5}|S_{\alpha,z_{0}}(\tfrac{1}{2}+\sigma-w+\i T_{N(l)})|\rbrace|dw\ll_\epsilon T_{N(l)}^{2+\epsilon}.\\
\end{equation}
which in turn implies (cf. \eqref{logbound})
\begin{equation}
\begin{split}
&\int_{T_{N(l)}-\i\sigma}^{T_{N(l)}}|\log|S_{\alpha,z_{0}}(\tfrac{1}{2}+\i\rho)|||d\rho|\\
= &\int_{0}^{\sigma}|\log|S_{\alpha,z_{0}}(\tfrac{1}{2}+\sigma-w+\i T_{N(l)})||dw\ll_\epsilon T_{N(l)}^{2+\epsilon}.
\end{split}
\end{equation}
\begin{flushright}
$\square$
\end{flushright}

Next we apply Proposition \ref{compbound} to derive the vanishing of the sequence of boundary terms \eqref{bdterms}.
\begin{thm}\label{van}
Let $\delta,\eta>0$ and $\lbrace T_{N(l)}\rbrace_{l}$ as above. Then for any $h\in H_{\sigma+\eta,\delta}$
\begin{equation}
\lim_{T_{N(l)}\to\infty}\int_{T_{N(l)}-\i\sigma}^{T_{N(l)}}h(\rho)\frac{S'_{\alpha,z_{0}}}{S_{\alpha,z_{0}}}{(\tfrac{1}{2}+\i\rho)}d\rho=0.
\end{equation}
\end{thm}
\begin{proof}
We follow the same lines as in Proposition \ref{compbound}. In exactly the same way as in the proof above we obtain the identity \eqref{id} for $h\in H_{\sigma+\eta,\delta}\subset H_{\sigma,\delta}$. So
\begin{equation}
\begin{split}
&\lim_{j\to\infty}\left\{\int_{T_{N(l)}-\i\sigma}^{T_{N(l)}}+\int_{-T_{N(l)}}^{-T_{N(l)}-\i\sigma}\right\}h(\rho)\frac{d}{d\rho}\log S_{\alpha,z_{0}}(\tfrac{1}{2}+\i\rho)d\rho\\
=&\lim_{j\to\infty}\int_{T_{N(l)}-\i\sigma}^{T_{N(l)}}\lbrace h'(\rho)-h'(-\bar{\rho})\rbrace\log|S_{\alpha,z_{0}}(\tfrac{1}{2}+\i\rho)|d\rho.
\end{split}
\end{equation}
The limit vanishes since for any small $\epsilon>0$
\begin{equation}
\begin{split}
&\int_{T_{N(l)}-\i\sigma}^{T_{N(l)}}|\lbrace h'(\rho)-h'(-\bar{\rho})\rbrace|\big|\log|S_{\alpha,z_{0}}(\tfrac{1}{2}+\i\rho)|\big||d\rho|\\
\ll\;&T_{N(l)}^{-2-\delta}\int_{T_{N(l)}-\i\sigma}^{T_{N(l)}}\big|\log|S_{\alpha,z_{0}}(\tfrac{1}{2}+\i\rho)|\big||d\rho|
\ll_\epsilon T_{N(l)}^{\epsilon-\delta}
\end{split}
\end{equation}
where we have used Propositions \ref{compbound}, and we observe that by Cauchy's theorem $h\in H_{\sigma+\eta,\delta}$ implies
\begin{equation}
|h'(\rho)|\ll(1+|\Re\rho|)^{-2-\delta}
\end{equation}
uniformly in $|\Im\rho|\leq\sigma$.
\end{proof}

As an application of Theorem \ref{van} we can now prove the trace formula. Let $h\in H_{\sigma,\delta}$ for any $\sigma>\tfrac{1}{2}$ and $\delta>0$. We will make a specific choice of $\sigma$ below. We take $T\to\infty$ in \eqref{trunc} and divide by $2$ to obtain
\begin{equation}\label{pretrace}
\begin{split}
\sum_{j=0}^{\infty}\lbrace h(\rho^{\alpha}_{j})-h(\rho_{j})\rbrace
=\;&-\frac{1}{2\pi}\int_{-\i\sigma-\infty}^{-\i\sigma+\infty}h'(\rho)\log S_{\alpha,z_{0}}(\tfrac{1}{2}+\i\rho)d\rho.\\
\end{split}
\end{equation}
where we have used Theorem \ref{van}, the existence of the limit and Weyl's law. 

In order to complete the proof of Theorem \ref{thm1} we need to expand the RHS of \eqref{pretrace} into an identity term and diffractive orbit terms. We define
\begin{equation}
G_{s}^{\Gamma\backslash\scrI}(z,w)=\sum_{\gamma\in\Gamma\backslash\scrI}G_{s}(z,\gamma w).
\end{equation}
In the first term of \eqref{pretrace} we can expand the logarithm in a power series
\begin{equation}\label{series}
\begin{split}
&\log\left[1+m\beta\psi(\tfrac{1}{2}+\i\rho)+\beta G^{\Gamma\backslash\scrI}_{\tfrac{1}{2}+\i\rho}\left(z_{0}, z_{0}\right)\right]\\
=\;&\log\left[1+m\beta\psi(\tfrac{1}{2}+\i\rho)\right]-\sum_{k=1}^{\infty}\frac{(-1)^{k}}{k}\left(\frac{\beta G^{\Gamma\backslash\scrI}_{\tfrac{1}{2}+\i\rho}(z_{0}, z_{0})}{1+m\beta\psi(\tfrac{1}{2}+\i\rho)}\right)^{k}.
\end{split}
\end{equation}
This series converges absolutely and uniformly for all $\rho\in\CC$ with $\Im\rho=-\sigma$ if $\sigma$ is sufficiently large. To see this consider the following estimate (where $B_k$ denotes the kth Bernoulli number)
\begin{equation}
\begin{split}
&\left|\beta^{-1}+m\psi(\tfrac{1}{2}+\sigma+\i t)\right|\\
\geq& -\left|\beta\right|^{-1}+m\left|\psi(\tfrac{1}{2}+\sigma+\i t)\right|\\
\geq&-\left|\beta\right|^{-1}+m\left|\log(\tfrac{1}{2}+\sigma+\i t)\right|\\
&-\frac{m}{2\left|\tfrac{1}{2}+\sigma+\i t\right|}
-m\sum_{n=1}^{\infty}\frac{B_{2n}}{n!}\left|\tfrac{1}{2}+\sigma+\i t\right|^{-2n}\\
\geq&-\left|\beta\right|^{-1}+m\log(\tfrac{1}{2}+\sigma)-\frac{m\pi}{2}\\
&-\frac{m}{1+2\sigma}-m\sum_{n=1}^{\infty}\frac{B_{2n}}{n!}(\tfrac{1}{2}+\sigma)^{-2n}.
\end{split}
\end{equation}
Since
\begin{equation}
\frac{1}{1+2\sigma}+\sum_{n=1}^{\infty}\frac{B_{2n}}{n!}\sigma^{-2n}=O(\sigma^{-1})
\end{equation}
we infer from the above for large enough $\sigma>\tfrac{1}{2}$ that
\begin{equation}
\left|\beta^{-1}+m\psi(\tfrac{1}{2}+\sigma+\i t)\right|\geq q(\beta)\log(\tfrac{1}{2}+\sigma)
\end{equation}
for some constant $0<q(\beta)<m$. Combining this with Lemma \ref{Greenbound} we obtain for $\Im\rho=-\sigma$ the estimate
\begin{equation}\label{bigbound}
\left|\frac{\beta G^{\Gamma\backslash\scrI}_{1/2+\i\rho}(z_{0},z_{0})}{1+m\beta\psi(\tfrac{1}{2}+\i\rho)}\right|
\leq\frac{q(\beta)^{-1}C(\Gamma,z_{0})}{(\tfrac{1}{2}+\sigma)^{1/2}\log(\tfrac{1}{2}+\sigma)}.
\end{equation}
We can now choose $\tilde{\sigma}(\beta)>\tfrac{1}{2}$ large enough such that
\begin{equation}\label{large1}
\frac{q(\beta)^{-1}C(\Gamma,z_{0})}{(\tfrac{1}{2}+\tilde{\sigma}(\beta))^{1/2}\log(\tfrac{1}{2}+\tilde{\sigma}(\beta))}<1
\end{equation}
and $\tilde{\sigma}(\beta)>\left|\Im\rho^{\alpha}_{0}\right|$. This choice ensures that the series \eqref{series} converges absolutely and uniformly for all $\rho\in\CC$ with $\Im\rho=-\tilde{\sigma}(\beta)$.
Let $\sigma(\beta)=\tilde{\sigma}(\beta)+\eta$ for some $\eta>0$ and $h\in H_{\sigma(\beta),\delta}$. We now have for the RHS of \eqref{pretrace}
\begin{equation}\label{tr2}
\begin{split}
&\frac{1}{2\pi\i}\int_{-\i\tilde{\sigma}-\infty}^{-\i\tilde{\sigma}+\infty}h(\rho)\frac{m\beta\psi'(\tfrac{1}{2}+\i\rho)}{1+m\beta\psi(\tfrac{1}{2}+\i\rho)}d\rho\\
+&\frac{1}{2\pi\i}\sum_{k=1}^{\infty}\frac{(-\beta)^{k}}{k}\sum_{\gamma_{1},\cdots,\gamma_{k}\in\Gamma\backslash\scrI}\int_{-\i\tilde{\sigma}-\infty}^{-\i\tilde{\sigma}+\infty}h'(\rho)\frac{\prod_{j=1}^{k}G_{1/2+\i\rho}(z_{0},\gamma_{j}z_{0})}{(1+m\beta\psi(\tfrac{1}{2}+\i\rho))^{k}}d\rho.
\end{split}
\end{equation}
Substituting the integral representation \eqref{intrep} of the free Green function and changing order of integration gives
\begin{equation}\label{tr3}
\begin{split}
&\frac{1}{2\pi\i}\int_{-\i\tilde{\sigma}-\infty}^{-\i\tilde{\sigma}+\infty}h(\rho)\frac{m\beta\psi'(\tfrac{1}{2}+\i\rho)}{1+m\beta\psi(\tfrac{1}{2}+\i\rho)}d\rho
+\frac{1}{2\pi\i}\sum_{k=1}^{\infty}\frac{(-\beta)^{k}}{k}\\ &\times\sum_{\gamma_{1},\cdots,\gamma_{k}\in\Gamma\backslash\scrI}\int_{-\i\tilde{\sigma}-\infty}^{-\i\tilde{\sigma}+\infty}\int_{\tau_{1}}^{\infty}\cdots\int_{\tau_{k}}^{\infty}\frac{h'(\rho)}{(1+m\beta\psi(\tfrac{1}{2}+\i\rho))^{k}}\\
&\times\prod_{j=1}^{k}\frac{\,e^{-\i\rho t_{j}}}{\sqrt{\cosh t_{j}-\cosh\tau_{j}}}\prod_{j=1}^{k}dt_{j}\,d\rho
\end{split}
\end{equation}
Note that we have the bound
\begin{equation}
\begin{split}
&\frac{|h'(\rho)|}{|1+m\beta\psi(\tfrac{1}{2}+\i\rho)|^{k}}\prod_{j=1}^{k}\left|\frac{\,e^{-\i\rho t_{j}}}{\sqrt{\cosh t_{j}-\cosh\tau_{j}}}\right|\\
\ll&\;e^{(\Im\rho-1/2)(t_{1}+\cdots+t_{k})}\,(1+\left|\Re\rho\right|)^{-2-\delta}(\log\left|\Re\rho\right|)^{-k}
\end{split}
\end{equation}
on the integrand. So we can exchange integration by Fubini's Theorem. Recall (cf. section 1.2.) that the only zero of $1+m\beta\psi(\tfrac{1}{2}+\i\rho)$ in the halflane $\Im\rho<0$ is given by $\rho=-\i v_{\beta}$. So we have by shifting the contour from $\Im\rho=-\tilde{\sigma}$ to $\Im\rho=-\nu$
\begin{equation}
g_{\beta,k}(t)=\frac{(-1)^{k}}{2\pi\i k}\int_{-\i\tilde{\sigma}-\infty}^{-\i\tilde{\sigma}+\infty}h'(\rho)\frac{e^{-\i\rho t}}{(1+m\beta\psi(\tfrac{1}{2}+\i\rho))^{k}}d\rho.
\end{equation}
So \eqref{tr3} equals
\begin{displaymath}
\begin{split}
&\frac{1}{2\pi\i}\int_{-\i\nu-\infty}^{-\i\nu+\infty}h(\rho)\frac{m\beta\psi'(\tfrac{1}{2}+\i\rho)}{1+m\beta\psi(\tfrac{1}{2}+\i\rho)}d\rho\\
+&\sum_{k=1}^{\infty}\beta^k \sum_{\gamma_{1},\cdots,\gamma_{k}\in\Gamma\backslash\scrI}\int_{\tau_{1}}^{\infty}\cdots\int_{\tau_{k}}^{\infty}\prod_{j=1}^{k}\frac{g_{\beta,k}(t_{1}+\cdots+t_{k})}{\sqrt{\cosh t_{j}-\cosh\tau_{j}}}\prod_{j=1}^{k}dt_j.
\end{split}
\end{displaymath}

We have the bound
\begin{equation}
|g_{\beta,k}(t)|\leq\frac{e^{-\tilde{\sigma} t}\int_{-\infty}^{\infty}|h'(\rho)|d\rho}{2\pi k\beta^{k}q(\beta)^{k}(\log(\tfrac{1}{2}+\tilde{\sigma}))^{k}}.
\end{equation}
We define the subset $\Gamma_{r,z_{0}}=\left\{\gamma\in\Gamma\backslash\scrI\mid d(z_{0},\gamma z_{0})<r\right\}$ of $\Gamma$ which is clearly finite because of discreteness of the group. For some $r>0$ we obtain the estimate
\begin{equation}
\begin{split}
&\sum_{\gamma_{1},\cdots,\gamma_{k}\in\Gamma_{r,z_{0}}}\int_{\tau_{1}}^{\infty}\cdots\int_{\tau_{k}}^{\infty}\frac{\left|g_{\beta,k}(t_{1}+\cdots+t_{k})\right|\prod_{n=1}^{k}dt_{n}}{\prod_{n=1}^{k}\sqrt{\cosh t_{n}-\cosh\tau_{n}}}\\
\ll&\;\beta^{-k}(q(\beta)\log(\tfrac{1}{2}+\tilde{\sigma}))^{-k}\left[\sum_{\gamma\in\Gamma_{r,z_{0}}}\int_{\tau_{\gamma}}^{\infty}\frac{e^{-\tilde{\sigma} t}dt}{\sqrt{\cosh t-\cosh\tau_{\gamma}}}\right]^{k}\\
\leq&\;\beta^{-k}(q(\beta)\log(\tfrac{1}{2}+\tilde{\sigma}))^{-k}\left[\sum_{\gamma\in\Gamma\backslash\scrI}\int_{\tau_{\gamma}}^{\infty}\frac{e^{-\tilde{\sigma} t}dt}{\sqrt{\cosh t-\cosh\tau_{\gamma}}}\right]^{k}\\
\leq&\;\beta^{-k}\left(\frac{q(\beta)^{-1}C(\Gamma,z_{0})}{(\tfrac{1}{2}+\tilde{\sigma})^{1/2}
\log(\tfrac{1}{2}+\tilde{\sigma})}\right)^{k}<\;\beta^{-k}
\end{split}
\end{equation}
where we follow the same lines as in the proof of Lemma \ref{Greenbound}. So independently of $r>0$ the sum over $k$ converges absolutely. Taking $r\to\infty$ we see that
$$\frac{1}{2\pi\i}\sum_{k=1}^{\infty}\frac{(-\beta)^{k}}{k}\sum_{\gamma_{1},\cdots,\gamma_{k}\in\Gamma\backslash\scrI}\int_{\tau_{1}}^{\infty}\cdots\int_{\tau_{k}}^{\infty}\frac{g_{\beta,k}(t_{1}+\cdots+t_{n})\prod_{k=1}^{n}dt_{n}}{\prod_{n=1}^{k}\sqrt{\cosh t_{n}-\cosh\tau_{n}}}$$
converges absolutely.

\section*{Acknowledgements}
I would like to thank my supervisor Jens Marklof for his guidance during the completion of this work. I am also grateful to Andreas Str\"ombergsson, Yiannis Petridis, Andy Booker and the anonymous referee for many helpful suggestions that have led to the improvement of this paper. 

\begin{appendix}

\section{Construction of the test function $h_\epsilon$}

In this section we give the proof of Lemma \ref{testf}. Let $\epsilon\in(0,1)$. We have to show that there exists $h_{\epsilon}\in H_{\sigma,\epsilon/2}$ such that
\begin{equation}\label{sym22}
\overline{h_{\epsilon}(\rho)}=h_{\epsilon}(\overline{\rho})
\end{equation}
and for some subsequence $\lbrace T_{N(j)}\rbrace_{j}\subset\lbrace T_{N}\rbrace_{N}$, $\lim_{j\to\infty}T_{N(j)}=\infty$ and $\rho\in[T_{N(j)},T_{N(j)}-\i\sigma]$ and $T_{N(j)}$ sufficiently large
\begin{equation}
|\Re h'_{\epsilon}(\rho)|=\Re h'_{\epsilon}(\rho)\gg T_{N(j)}^{-2-\epsilon}
\end{equation}
uniformly in $\rho$.

\begin{proof}
Let $\omega=\epsilon/10$. Since $\lim_{N\to\infty}T_{N}=\infty$, we can pick a subsequence $\lbrace T_{N(k)}\rbrace_{k=1}^{\infty}\subset\lbrace T_{N}\rbrace_{N=1}^{\infty}$ such that 
\begin{equation}
2^{n(k)-1}\leq T_{N(k)}, \;T_{N(k)}+T_{N(k)}^\omega\leq 2^{n(k)+2}
\end{equation} 
for some integer $n(k)\geq1$ with $n(k+1)>n(k)+3$ for all $k\geq1$. Now let $\sigma_{0}>\sigma$ and consider the test function
\begin{equation}
\begin{split}
h_{\epsilon}(\rho)=\sum_{k=1}^{\infty}2^{-n(k)(2+\epsilon/2)}\bigg\{
& \frac{1}{(\rho-T_{N(k)}-T_{N(k)}^\omega-\i\sigma_{0})^{4}}\\
&+\frac{1}{(\rho-T_{N(k)}-T_{N(k)}^\omega+\i\sigma_{0})^{4}}\\
& +\frac{1}{(\rho+T_{N(k)}+T_{N(k)}^\omega+\i\sigma_{0})^{4}}\\
&+\frac{1}{(\rho+T_{N(k)}+T_{N(k)}^\omega-\i\sigma_{0})^{4}}\bigg\}.
\end{split}
\end{equation}
By construction the property \eqref{sym22} is fulfilled. $h_{\epsilon}$ is even and analytic in the strip $\Im\rho\leq\sigma$. We will show that it also satisfies $h_{\epsilon}(\rho)\ll(1+|\Re\rho|)^{-2-\epsilon/2}$ uniformly in the strip $\Im\rho\leq\sigma$. Because of evenness we only have to prove the bound for $\Re\rho>0$. We have
\begin{equation}
\begin{split}
&|h_{\epsilon}(\rho)|\leq2\sum_{k=1}^{\infty}2^{-n(k)(2+\epsilon/2)}\\
&\left\{\frac{1}{|\rho-T_{N(k)}-T_{N(k)}^\omega+\i\sigma_{0}|^{4}}+\frac{1}{|\rho+T_{N(k)}+T_{N(k)}^\omega+\i\sigma_{0}|^{4}}\right\}.
\end{split}
\end{equation}
We will estimate the two parts separately. For the second sum we have
\begin{equation}\label{secondsum}
\begin{split}
&\sum_{k=1}^{\infty}2^{-n(k)(2+\epsilon/2)}\frac{1}{|\rho+T_{N(k)}+T_{N(k)}^\omega+\i\sigma_{0}|^{4}}\\
&\leq\sum_{k=1}^{\infty}2^{-n(k)(2+\epsilon/2)}\frac{1}{(\Re\rho+T_{N(k)}+T_{N(k)}^\omega)^{4}}\\
&\leq|\Re\rho|^{-4}\sum_{k=1}^{\infty}2^{-k(2+\epsilon/2)}.
\end{split}
\end{equation}
The first sum is more difficult to estimate. Since $h_{\epsilon}$ is analytic in the strip $\Im\rho\leq\sigma$ it suffices to prove the bound for $\Re\rho\geq1$. So there exists an integer $j\geq1$ such that $2^{j-1}\leq\Re\rho\leq2^{j}$. We obtain
\begin{equation}
\begin{split}
&\sum_{k=1}^{\infty}2^{-n(k)(2+\epsilon/2)}\frac{1}{|\rho-T_{N(k)}-T_{N(k)}^\omega+\i\sigma_{0}|^{4}}\\
\leq&\sum_{n(k)\neq j-3,\cdots,j+1}2^{-n(k)(2+\epsilon/2)}\frac{1}{|\Re\rho-T_{N(k)}-T_{N(k)}^\omega|^{4}}\\
&+\frac{1}{|\Im\rho+\sigma_{0}|^{4}}\sum_{i=1}^5 2^{-(j-4+i)(2+\epsilon/2)}.
\end{split}
\end{equation}
Now, since $\Re\rho\geq2^{j-1}$ and $|\Im\rho|\leq\sigma<\sigma_{0}$,
\begin{equation}
\frac{1}{|\Im\rho+\sigma_{0}|^{4}}\sum_{i=1}^5 2^{-(j-4+i)(2+\epsilon/2)}\ll|\Re\rho|^{-2-\epsilon/2}.
\end{equation}
Next we estimate the sum. First observe that for $n(k)\leq j-4$
\begin{equation}
|T_{N(k)}+T_{N(k)}^\omega-\Re\rho|\geq2^{j-1}-2^{n(k)+2}=2^{j-1}(1-2^{n(k)-j+3})\geq2^{j-2}
\end{equation}
and for $n(k)\geq j+2$
\begin{equation}
|T_{N(k)}+T_{N(k)}^\omega-\Re\rho|\geq2^{n(k)-1}-2^{j}=2^{j}(2^{n(k)-j-1}-1)\geq2^{j},
\end{equation}
which implies
\begin{equation}\label{estrem}
\begin{split}
&\sum_{k\neq j-3,\cdots,j+1}2^{-n(k)(2+\epsilon/2)}\frac{1}{|\Re\rho-T_{N(k)}-T_{N(k)}^\omega|^{4}}\\
&\leq 2^{8-4j}\sum_{k\neq j-3,\cdots,j+1}2^{-n(k)(2+\epsilon/2)}\\
&\leq 2^{8-4j}\sum_{k=1}^{\infty}2^{-k(2+\epsilon/2)}\\
&\leq 2^{4}|\Re\rho|^{-4}\sum_{k=1}^{\infty}2^{-k(2+\epsilon/2)}.
\end{split}
\end{equation}
Next we prove for $\rho\in[T_{N(j)},T_{N(j)}-\i\sigma]$ the lower bound $\Re h_{\epsilon}'(\rho)\gg T_{N(j)}^{-2-\epsilon}$ uniformly in $\rho$. We have
\begin{equation}\label{testder}
\begin{split}
h'_{\epsilon}(\rho)=-4\sum_{k=1}^{\infty}2^{-n(k)(2+\epsilon/2)}\bigg\{
& \frac{1}{(\rho-T_{N(k)}-T_{N(k)}^\omega-\i\sigma_{0})^{5}}\\
&+\frac{1}{(\rho-T_{N(k)}-T_{N(k)}^\omega+\i\sigma_{0})^{5}}\\
& +\frac{1}{(\rho+T_{N(k)}+T_{N(k)}^\omega+\i\sigma_{0})^{5}}\\
&+\frac{1}{(\rho+T_{N(k)}+T_{N(k)}^\omega-\i\sigma_{0})^{5}}\bigg\}.
\end{split}
\end{equation}
Since
\begin{equation}\label{realp}
\begin{split}
&\Re\left[\frac{1}{(\rho-T_{N(k)}-T_{N(k)}^\omega\pm\i\sigma_{0})^{5}}\right]\\
=&\frac{(\Re\rho-T_{N(k)}-T_{N(k)}^\omega)^{5}}{((\Re\rho-T_{N(k)}-T_{N(k)}^\omega)^{2}+(\Im\rho\pm\sigma_{0})^{2})^{5}}\\
&-\frac{10(\Re\rho-T_{N(k)}-T_{N(k)}^\omega)^{3}(\Im\rho\pm\sigma_{0})^{2}}{((\Re\rho-T_{N(k)}-T_{N(k)}^\omega)^{2}+(\Im\rho\pm\sigma_{0})^{2})^{5}}\\
&+\frac{5(\Re\rho-T_{N(k)}-T_{N(k)}^\omega)(\Im\rho\pm\sigma_{0})^{4}}{((\Re\rho-T_{N(k)}-T_{N(k)}^\omega)^{2}+(\Im\rho\pm\sigma_{0})^{2})^{5}}
\end{split}
\end{equation}
we have for $\rho\in[T_{N(j)},T_{N(j)}-\i\sigma]$
\begin{equation}
\begin{split}
&\sum_{k=1}^{\infty}2^{-n(k)(2+\epsilon/2)}\Re\left[\frac{1}{(\rho-T_{N(k)}-T_{N(k)}^\omega\pm\i\sigma_{0})^{5}}\right]\\
=\;&-2^{-n(j)(2+\epsilon/2)}\frac{T_{N(j)}^{5\omega}-10T_{N(j)}^{3\omega}(\Im\rho\pm\sigma_{0})^{2}+5 T_{N(j)}^\omega(\Im\rho\pm\sigma_{0})^{4}}{(T_{N(j)}^{2\omega}+(\Im\rho\pm\sigma_{0})^{2})^{5}}\\
&+\sum_{k\neq j}2^{-n(k)(2+\epsilon/2)}\Re\left[\frac{1}{(\rho-T_{N(k)}-T_{N(k)}^\omega\pm\i\sigma_{0})^{5}}\right].
\end{split}
\end{equation}
Now it follows from \eqref{realp} and $\Re\rho=T_{N(j)}$ that for $k\neq j$ 
\begin{equation}
\left|\Re\left[\frac{1}{(\rho-T_{N(k)}-T_{N(k)}^\omega\pm\i\sigma_{0})^{5}}\right]\right|\ll|T_{N(j)}-T_{N(k)}-T_{N(k)}^\omega|^{-5}
\end{equation}
and thus
\begin{equation}
\begin{split}
&\left|\sum_{k\neq j}2^{-n(k)(2+\epsilon/2)}\Re\left[\frac{1}{(\rho-T_{N(k)}-T_{N(k)}^\omega\pm\i\sigma_{0})^{5}}\right]\right|\\
\ll&\sum_{k\neq j}2^{-n(k)(2+\epsilon/2)}|T_{N(j)}-T_{N(k)}-T_{N(k)}^\omega|^{-5}.
\end{split}
\end{equation}
To bound this we distinguish two cases. First assume $j>k$. Then $n(j)>n(k)+3$. So we have $T_{N(k)}+T_{N(k)}^\omega\leq 2^{n(k)+2}<2^{n(j)-1}\leq T_{N(j)}$. This implies
\begin{equation}
\begin{split}
&|T_{N(k)}+T_{N(k)}^\omega-T_{N(j)}|\geq 2^{n(j)-1}-2^{n(k)+2}\\
=\;&2^{n(j)-1}(1-2^{n(k)-n(j)+3})\geq2^{n(j)-2}.
\end{split}
\end{equation}
Now assume $j<k$. Then $n(k)>n(j)+3$. In this case we have $T_{N(j)}\leq2^{n(j)+2}<2^{n(k)-1}\leq T_{N(k)}+T_{N(k)}^\omega$. This implies
\begin{equation}
\begin{split}
&|T_{N(k)}+T_{N(k)}^\omega-T_{N(j)}|\geq2^{n(k)-1}-2^{n(j)+2}\\
=\;&2^{n(j)+2}(2^{n(k)-n(j)-3}-1)>2^{n(j)+2}.
\end{split}
\end{equation}
It follows from $2^{n(j)-1}<T_{N(j)}<2^{n(j)+2}$ that
\begin{equation}\label{knotjbound}
\begin{split}
&\sum_{k\neq j}2^{-n(k)(2+\epsilon/2)}|T_{N(j)}-T_{N(k)}-T_{N(k)}^\omega|^{-5}\\
&\leq\;2^{10-5n(j)}\sum_{k\neq j}2^{-n(k)(2+\epsilon)}\\
&\ll T_{N(j)}^{-5}.
\end{split}
\end{equation}
So for large $T_{N(j)}$ (recall $\Re\rho=T_{N(j)}$) positivity of the leading term implies that the real part of $h'_{\epsilon}(\rho)$ is positive. Note (in view of \eqref{testder}) that the contributions from the remaining two terms are of order $O(T_{N(j)}^{-5})$ by the same argument as in \eqref{secondsum}:
\begin{equation}\label{est}
\begin{split}
&\qquad\Re h'_\epsilon(\rho)\\
&=-4\sum_{k=1}^{\infty}2^{-n(k)(2+\epsilon/2)}\Re\Bigg[\frac{1}{(\rho-T_{N(k)}-T_{N(k)}^\omega-\i\sigma_{0})^{5}}\\
&\qquad\qquad\qquad\qquad\qquad\qquad+\frac{1}{(\rho-T_{N(k)}-T_{N(k)}^\omega+\i\sigma_{0})^{5}}\Bigg]\\
&\qquad+O(T_{N(j)}^{-5})
\end{split}
\end{equation}
and by \eqref{knotjbound} this equals
\begin{equation}
\begin{split}
&2^{2-n(j)(2+\epsilon/2)}\frac{T_{N(j)}^{5\omega}-10T_{N(j)}^{3\omega}(\Im\rho-\sigma_{0})^{2}+5 T_{N(j)}^\omega(\Im\rho-\sigma_{0})^{4}}{(T_{N(j)}^{2\omega}+(\Im\rho-\sigma_{0})^{2})^{5}}\\
&+2^{2-n(j)(2+\epsilon/2)}\frac{T_{N(j)}^{5\omega}-10T_{N(j)}^{3\omega}(\Im\rho+\sigma_{0})^{2}+5 T_{N(j)}^\omega(\Im\rho+\sigma_{0})^{4}}{(T_{N(j)}^{2\omega}+(\Im\rho+\sigma_{0})^{2})^{5}}\\
&\qquad+O(T_{N(j)}^{-5})\\
&\gg T_{N(j)}^{-2-\epsilon/2-5\omega}=T_{N(j)}^{-2-\epsilon}
\end{split}
\end{equation}
uniformly in $\rho$. It follows that for $T_{N}(j)$ large and $\rho\in[T_{N(j)},T_{N(j)}-\i\sigma]$ we have $|\Re h'_{\epsilon}(\rho)|=\Re h'_{\epsilon}(\rho)\gg T_{N(j)}^{-2-\epsilon}$ uniformly in $\rho$.
\end{proof}

\end{appendix}

\bibliographystyle{plain}

\end{document}